\documentclass[11pt,final]{article}
\usepackage{fullpage}

\usepackage{amsmath}
\usepackage{amsthm}
\usepackage{url}
\usepackage{nicefrac}
\usepackage{graphicx}
\usepackage{float}
\usepackage{thmtools}
\usepackage{thm-restate}
\usepackage{wrapfig}

\usepackage{comment}
\usepackage{amsfonts}
\usepackage{natbib}
\usepackage{enumerate}

\usepackage[colorlinks=true]{hyperref}

\hypersetup{
	linkcolor=[rgb]{0.4,0,0.6},
	citecolor=[rgb]{0, 0.4, 0},
	urlcolor=[rgb]{0.6, 0, 0}
}

\usepackage[capitalise]{cleveref}
\crefformat{footnote}{#2\footnotemark[#1]#3}

\newtheorem{theorem}{Theorem}[section]
\crefname{theorem}{Theorem}{Theorems}
\newtheorem*{theorem*}{Theorem}
\newtheorem*{proposition*}{Proposition}
\newtheorem*{question*}{Main Question}
\newtheorem{claim}[theorem]{Claim}
\crefname{claim}{Claim}{Claims}
\newtheorem{proposition}[theorem]{Proposition}
\crefname{proposition}{Proposition}{Propositions}
\newtheorem{corollary}[theorem]{Corollary}
\crefname{corollary}{Corollary}{Cropositions}
\newtheorem{lemma}[theorem]{Lemma}
\crefname{lemma}{Lemma}{Lemmas}
\newtheorem{observation}[theorem]{Observation}
\crefname{observation}{Observation}{Observations}
\theoremstyle{plain}

\crefname{question}{Main Question}{Main Questions}
\newtheorem{example}[theorem]{Example}
\crefname{example}{Example}{Examples}
\newtheorem{definition}[theorem]{Definition}
\crefname{definition}{Definition}{Definitions}

\newcommand{\bs}{\mathbf s}
\newcommand{\bv}{\mathbf v}
\newcommand{\bb}{\mathbf b}
\newcommand{\snoi}{{\mathbf s}_{-i}}

\newcommand{\SW}{\textsc{SW}}
\newcommand{\opt}{\normalfont\textsc{Opt}}
\newcommand{\tildeopt}{\widetilde{\opt}}
\newcommand{\tildem}{\widetilde{m}}

\newcommand{\eq}{\normalfont\textsc{Eq}}
\newcommand{\epoa}{\normalfont\textsc{EP-PoA}}
\newcommand{\npoa}{\normalfont\textsc{NE-PoA}}
\newcommand{\bpoa}{\normalfont\textsc{B-PoA}}

\newcommand{\epe}{\textit{EPE}}
\newcommand{\pne}{\textit{PNE}}
\newcommand{\bne}{\textit{BNE}}
\newcommand{\self}{\mathsf{SELF}}
\newcommand{\other}{\mathsf{OTHER}}

\newcommand{\s}{\mathbf{s}}
\newcommand{\bt}{\mathbf{t}}
\newcommand{\avec}{\mathbf{a}}
\newcommand{\bids}{\mathbf{b}}
\newcommand{\E}{\mathbb{E}}
\newcommand{\Exp}{\operatorname{\E}}
\newcommand{\I}{\mathbb{I}}
\newcommand{\F}{\mathcal{F}}
\newcommand{\argmax}{\operatorname*{argmax}}
\newcommand{\Var}{\operatorname{Var}}
\newcommand{\Cov}{\operatorname{Cov}}

\newcommand{\MyFrame}[1]{\noindent \framebox[\textwidth]{ \begin{minipage}{0.97\textwidth} #1 \end{minipage}}}%

\allowdisplaybreaks

\begin{document}
\title{Price of Anarchy of Simple Auctions with Interdependent Values%
\thanks{
	The work of A. Eden, M. Feldman and O. Zviran  was partially supported by the European Research Council (ERC) under the European Union's Horizon 2020 research and innovation program (grant agreement No. 866132), and by the Israel Science Foundation (grant number 317/17).
	The work of I. Talgam-Cohen was supported by the ISRAEL SCIENCE FOUNDATION (grant No. 336/18) and by the Taub Family Foundation.
}
}
\author{Alon Eden%
\thanks{%
    {Harvard University(\url{aloneden@seas.harvard.edu})}}
\and Michal Feldman%
\thanks{%
    {Tel Aviv University (\url{michal.feldman@cs.tau.ac.il})}}
\and Inbal Talgam-Cohen%
\thanks{%
    {Technion, Israel Institute of Technology (\url{italgam@cs.technion.ac.il})}}
\and Ori Zviran%
\thanks{%
    {Tel Aviv University(\url{orizviran@mail.tau.ac.il})}}
}

\maketitle
\begin{abstract}
	

We expand the literature on the price of anarchy (PoA) of simultaneous item auctions by considering settings with correlated values; we do this via the fundamental economic model of {\em interdependent values (IDV)}.
It is well-known that in multi-item settings with private values, correlated values can lead to bad PoA, which can be polynomially large in the number of agents~$n$.
In the more general model of IDV, we show that the PoA can be polynomially large even in single-item settings.
On the positive side, we identify a natural condition on information dispersion in the market, termed \emph{$\gamma$-heterogeneity}, which enables good PoA guarantees. 
Under this condition, we show that for single-item settings, the PoA of standard mechanisms degrades gracefully with~$\gamma$.
For settings with $m>1$ items we show a separation between two domains: 
If $n \geq m$, we devise a new simultaneous item auction with good PoA (with respect to $\gamma$), under limited information asymmetry.
To the best of our knowledge, this is the first positive PoA result for correlated values in multi-item settings.
The main technical difficulty in establishing this result is that the standard tool for establishing PoA results --- the smoothness framework --- is unsuitable for IDV settings, and so we must introduce new techniques to address the unique challenges imposed by such settings.
In the domain of $n \ll m$, we establish impossibility results even for surprisingly simple scenarios.

%
%
%
%
%
%

\end{abstract}

\section{Introduction}
\label{sec:intro}
We study simple and practical mechanisms for selling heterogeneous items to unit-demand agents,%
\footnote{Where an agent has a value for each item, and her value for a bundle is the maximal value for a bundle's item.}
with the goal of maximizing social welfare. 
While much of the literature focuses on truthful mechanisms (for which it is in the agents' best interest to report their true values for the items), 
such mechanisms are often complicated, computationally and cognitively demanding, and must be run in a centralized manner~\citep{Dobzinski11,li2017obviously,ausubel2006lovely}. In many real-life settings like e-commerce, 
the mechanisms used in practice are simple and run in a distributed manner, but are non-truthful. A prime example is the auctions run by e-Bay, where each item is sold separately in a second-price auction, so that an agent interested in winning at most one auction may be better off reporting lower than her true values, to avoid multiple wins.

Motivated by such settings,~\cite{CKS16} pioneered the study of \textit{simultaneous item auctions}, where a single-item auction is run for each item separately. Since such auctions are non-truthful, their performance is measured by their \textit{Bayesian Price of Anarchy} (\bpoa) ---  	 the ratio between the optimal welfare and the welfare guarantee of the worst equilibrium.
Christodoulou et al.~and many follow-ups show that this simple format achieves near-optimal welfare in combinatorial settings when the auction in use is the first- or second-price auction (see~\cite{RoughgardenST17} for an overview of results).
Importantly, all prior works assume independence of different buyers' valuations, since otherwise the PoA might be polynomial in the number of agents~\citep{BhawalkarR11,FeldmanFGL13,Roughgarden14}. 

There are many settings in which the independence assumption is unrealistic. For instance, if an item being auctioned has the potential of being resold in the future, this potential factors into agents' values for the item, creating dependence on how others value it. Such dependence is formally captured by the \textit{interdependent values (IDV)} model~\citep{MilgromW82}, which recently received a surge of interest following the recent awarding of the Nobel prize in economics to Milgrom and Wilson for their work on auction theory.
In this model, the correlation of values stems directly from one of the most fundamental aspects of a market --- the way in which information is dispersed among agents. In more detail, each agent has a privately-known \emph{signal}, which captures her partial knowledge about the items for sale. Her value for each item is a function of all the information on the market related to this item; that is, of her own signal as well as the signals of all other agents (which are unknown to her). Since the valuations of different agents depend on the same signals, values are correlated.

IDV settings have been studied extensively in the economic literature since \citet{MilgromW82}, and in the computer science literature \citep{ItoP06,ConstantinIP07,RobuPIJ13,CFK,RoughgardenT16,EdenFFGK19}; see \citet{krishna2009auction} for an overview and Section~\ref{sec:related-work} for additional related work.
Realistic scenarios captured by this model include common value auctions~\citep{klemperer1998auctions}, mineral rights auctions~\citep{wilson1969communications} and resale~\citep{myerson1981optimal,RoughgardenT16}. We use the latter as a running example throughout the paper:

\begin{example}[Resale model] \label{example:running}
	Let $s_i$ be the private signal of agent $i\in [n]$, distributed uniformly between $0$ and $1$. This signal captures the appreciation of agent $i$ for the item for sale. Agent $i$'s  \textit{valuation function} is $v_i(\s)= s_i+\beta\sum_{j\ne i} s_j$ for some parameter $\beta\in(0,1)$. That is, the agent also takes into account the others' assessment of the item
	(possibly since she plans to resell the item eventually). 
\end{example}

Our main question is: 

\begin{question*}
	What is the B-PoA of simultaneous item auctions with IDV?
\end{question*}

The bounds we establish partially answer an open question of \cite{RoughgardenST17} regarding natural, economically-meaningful forms of correlation for which the B-PoA of simultaneous item auctions is bounded.

\subsection{Challenges}
\label{sec:challenges}

Several challenges arise when approaching the task of bounding the B-PoA in IDV.
As intuition suggests, the domain of problems that can be modeled by interdependent values is very large, and
includes the mostly-studied setting of private values, correlated or uncorrelated.
Due to the wide scope, it is unsurprising that in full generality there is no hope to achieve a good bound for the B-PoA. Our lower bounds formalize the need for assumptions in order to get a bound independent of the number of agents.

A standard assumption in the interdependent model is {\em single-crossing (SC)}. Intuitively it means that the information possessed by an agent has more influence on her own value than on others' values. SC enables many positive results in settings with interdependence \citep{Ausubel00,Maskin96,CFK}.
For example, in a single item setting, the truthful (generalized) Vickrey auction achieves optimal welfare only with SC.
However, SC alone is insufficient for ensuring good B-PoA results of simple auctions, even with a single item. 
The following proposition shows that any mechanism that allocates the item to the agent with the highest value is prone to a large degree of social inefficiency even with SC.
\begin{proposition}
	\label{clm:GVA:POA}
	There exists a single-item, $n$-bidder setting satisfying SC such that the PoA of every auction that allocates the item to the highest-valued bidder is $\Omega(n)$.%
	\footnote{This holds even under the standard no-overbidding assumption (see Section~\ref{sec:prelim}) and under additional conditions such as submodularity over signals \cite[]{EdenFFGK19}.}
\end{proposition}

While the full proof is deferred to the Appendix~\ref{appx:challenges}, the intuition behind this negative result is instructive for identifying necessary assumptions for positive results: Consider $n$ bidders whose signals belong to subsets $S_1 = \{1\}$, $S_2 = \{1\}$, and $S_i = [0, 1]$ for $i \geq 3$, respectively. Consider the valuation profile
$
v_1 = \sum_{i\in [n]} s_i +\epsilon$,
$v_2 = 2(s_2+\epsilon)$, and $v_i = s_i\label{eq:bad-example1}
$ for $i\ge3$,
where $\epsilon>0$ is arbitrarily small.
In this setting, the signals of bidders $i \geq 3$ have a significant effect on bidder~1's value but have no effect on bidder 2's value. In a scenario where these bidders have high signals, bidder~1's value is significantly higher than bidder 2's value. 
However, bidders $i \geq 3$ never win the item, since bidder $1$ out-values them for every signal profile. Thus, they may as well report low signals in equilibrium, resulting in an outcome where bidder~2 wins the item. 
A similar result to \cref{clm:GVA:POA} holds for every deterministic truthful (ex-post incentive compatible) mechanism.

\paragraph{Incentive compatibility.}
In the context of interdependence, incentive compatibility (IC) and individual rationality (IR) are defined \emph{ex-post}, i.e., it is in every bidder's best interest to participate and report her true signal \emph{given that all other bids are also truthful} (bidding truthfully is an ex-post equilibrium of the mechanism -- see \cref{sec:equilibrium-notions}).

\begin{definition}[Ex-post IC-IR]
	A single-item mechanism is \emph{ex-post IC-IR} if for every signal profile $\bs$, bidder $i$ and bid~$b_i$, 
	\begin{eqnarray*}
	x_i(\bs)v_i(\bs)-p_i(\bs) \ge \max\{x_i(b_i,\snoi)v_i(\bs)-p_i(b_i,\snoi),0\}.
	\end{eqnarray*}
\end{definition}

Ex-post IC-IR mechanisms can be characterized as mechanisms with a monotone allocation rule \cite[e.g.][]{RoughgardenT16}: for every bidder $i$ and partial bid profile $\bb_{-i}$, the allocation rule $x_i(b_i,\bb_{-i})$ is (weakly) increasing in $b_i$. The monotone allocation rule is coupled with a payment rule uniquely determined by the allocation rule. For deterministic mechanisms the winner $i$ is charged by his critical bid. 
For this class of mechanisms we show the following (see Appendix~\ref{appx:challenges}).

\begin{proposition}
	\label{pro:POA-EPIC}
	There exists a single-item, $n$-bidder setting, satisfying single-crossing, such that the PoA of every deterministic ex-post IC-IR mechanism is $\Omega(\sqrt{n})$ (even under no-overbidding).
\end{proposition}


%

\subsection{Our Results}

We study simple mechanisms in which the agents report their signals to the auction(s). 
All our positive results hold with respect to Bayesian equilibrium, which is the strongest guarantee, and thus propagate to other equilibrium notions (see Observation~\ref{obs:equil-hierarchy}). Moreover, we show that pure equilibria need not exist even in a single-item setting (see \cref{pro:no-pure-equil}), which further motivates our choice of benchmark.%
\footnote{In contrast, for simultaneous auctions with private values, it is known that a PNE exists even for general classes of valuations (e.g.~\cite{CKS16}).}
Since we are tackling new uncharted territory, our starting point is a single item, where we already face new challenges that do not exist in private values settings. 
\paragraph{Single Item.}
We consider natural auction formats like first-price and second-price. As discussed above, the PoA of such auctions inherently depends on the number of agents. 
In the proofs of Propositions~\ref{clm:GVA:POA} and~\ref{pro:POA-EPIC}, a single bidder is highly influenced by the signals of others, while all other bidders have private values. Thus, the effect of bidder signals on different bidders may vary greatly. We parameterize the extent of this variation by $\gamma$ (\cref{def:influence}) and establish positive results that depend on $\gamma$. 

Our main positive result for a single item applies to generalized Vickrey auction (GVA) --- a natural generalization of the Vickrey auction to interdependent values~\cite[]{Maskin96,Ausubel00}. We show that the $\bpoa$ of GVA under SC is $\gamma$ (see Theorem~\ref{prop:upper-bound:gamma-bound+c-sc:GVA}), and this is tight (see Theorem~\ref{pro:gamma-bound-csc:POA-close-gamma-c}). When considering a relaxed notion of SC called $c$-SC, we get a bound of $\max\{\gamma,c\}+1$ for both GVA and second-price auctions (Theorems~\ref{prop:upper-bound:gamma-bound+c-sc:GVA} and~\ref{thm:upper-bound:gamma-bound+c-sc:2PA}), which is almost tight (see Theorems~\ref{pro:gamma-bound-csc:POA-close-gamma-c} and \ref{pro:gamma-bound-csc:POA-close-gamma-c:2PA}). 
A non-trivial implication of this result is that in every IDV setting where  $\gamma=c=1$ it holds that every equilibrium is fully efficient. For example, this is the case in our running example of the resale model (\cref{example:running}).
\paragraph{Multiple Items.}

We consider a combinatorial setting with heterogeneous items and unit-demand valuations. We allow for signals to be multi-dimensional, a notoriously hard setting in the IDV literature~\cite[]{dasgupta2000efficient,jehiel2001efficient}.  Our results for this case are more nuanced. First of all, we show that in general, one cannot approximate the optimal welfare for a very natural class of mechanisms which includes the ones we reason about in this paper. In the example showing this, the value of all bidders originates from a single bidder's signal; that is, this setting suffers from extreme information asymmetry. Therefore, we provide results for settings with limited-information asymmetry, like our running example (for a formal definition of the condition, see \cref{sub:sufficient-cond}).

We then show a separation between two domains: the case where there are more bidders than items ($n\geq m$), and the case where there are many more items than bidders. For the case $n\geq m$, we consider a mechanism that runs a second-price auction for each item considering every bidder's value to be her valuation valuated at her bid, zeroing out other bidders' bids. We refer to this mechanism as the \emph{simultaneous privatized} second-price auction.
One subtlety that arises in our setting is the need to enable agents to explicitly express willingness to participate in auctions for individual items. Our main positive result for multiple items is that for $\gamma$-heterogeneous unit-demand bidders, and under sufficiently limited  information asymmetry, the $\bpoa$ of simultaneous privatized second-price auctions is $O(\gamma^2)$ (under $c$-SC, we get a bound of $O(\max\{\gamma,c\}^2)$ --- see Theorem~\ref{thm:multi-pos} for details).

Our proof diverges from the smoothness framework, as the typical proof requires reasoning about a deviation of bidders who know their value for an item, which is not the case with interdependence. Our proof decomposes the welfare into two terms and bounds the performance of our auction for each one separately (for more details, see Section~\ref{sec:multiple-pos}).


For the case $m\gg n$, we provide a strong negative result. We consider the simplest setting one can imagine: there are $n^2$ items, and bidders have a common value for each item, which is based on the sum of signals of all bidders for the item ($v_{i\ell}(\s)=1+\sum_j s_{j\ell}$ for every item $\ell$). All signals (for any bidder/item pair) are sampled i.i.d. 
This example demonstrates the difficulty of coordination in our setting: For every item, the bidders' values are equal, and thus the difficulty is in identifying the valuable items, despite the fact that the information is distributed among bidders. A priori, all items seem the same, but the signals' distribution is set up in a way that gives rise to a ball-and-bins-type phenomenon --- the expected value of each item is a constant, but the maximal value of $n$ items is $\Theta(\ln n/\ln\ln n)$. It follows that $\opt=\Theta(n\ln n/\ln\ln n)$.  
On the other hand, by the standard no-overbidding assumption (which is crucial for our positive results in the $n\leq m$ regime), the number of items a bidder bids on must be bounded, which in turn upper-bounds the expected welfare of the bidder. As a result, we show that for the settings discussed, for every simultaneous item auction and every equilibrium $\sigma$, $\opt/\eq(\sigma)=\tilde{\Omega}(\log n)$ (see Section~\ref{sec:multiple-neg} for the full details).%
\footnote{$\tilde{\Omega}$ hides $o(\log n)$ terms.}
Note that this implies that the Price of Stability is $\tilde{\Omega}(\log n)$, which is a stronger negative result than previously known.

\subsection{Related Work}
\label{sec:related-work}

\paragraph{Smoothness Framework.}
Many of the papers analyzing PoA of auctions make use of the celebrated smoothness framework~\cite[]{Roughgarden15}, and its adaptation to auctions and games of incomplete information~\cite[]{Roughgarden15incomplete, SyrgkanisT13}. More on this can be found in the survey of~\cite{RoughgardenST17}.

\paragraph{Simultaneous Auctions.} Following~\cite{CFK}, simultaneous auctions were studied by many follow-ups, where the auction ran simultaneously is either the second-price  (Vickrey)~\cite[]{BhawalkarR11,FuKL12,Roughgarden15incomplete}, or the first-price auction~\cite[]{HassidimKMN11,SyrgkanisT13,ChristodoulouKST16,FeigeFIILS15}.
These efforts culminated in showing that simultaneous second-price (under a no-overbidding condition) and first-price auctions achieve PoA of $4$ and $2$, respectively, for subadditive valuations, the most general class of complement-free valuations~\cite[]{FeldmanFGL13}. These results require values to be \emph{uncorrelated}. 	

\paragraph{Simple Auctions with Correlated Values.}
\citet{SyrgkanisT13} show that for first-price auction of a single item, the price of anarchy is $1-1/e$, even for correlated value distributions, and this is tight~\cite[]{Syrgkanis14}. When the values are uncorrelated, \citet{HoyTW18} show that this bound can be improved. 
The work of \citet{LucierL11} applies to a single item with correlated values sold through a second-price auction (and more generally, to multiple units sold through the \emph{GSP} auction\footnote{An auction format used in practice to allocate sponsored-search ad slots. GSP stands for \emph{Generalized Second Price}, not to be confused with the (quite different) generalized Vickrey auctions we study in this paper!}%
). Under a no-overbidding assumption, every BNE of the auction achieves at least $1/4$ of the maximum expected welfare. This no-overbidding assumption or variants thereof are required for all known positive PoA results related to second-price auctions \cite[see][Section 8.1]{RoughgardenST17}. 
The result of \citet{LucierL11} is by a strengthened variant of smoothness called \emph{semi-smoothness} or \emph{smoothness with private deviations}. The semi-smoothness framework was further developed by \citet{CaragiannisKKKLLT15} where the bound of $1/4$ was improved to $1/2.927$.

\paragraph{Other Objectives and Auction Formats.}
\citet{HartlineHT14} and \citet{AzarFGR17} consider the simultaneous item auctions in the context of maximizing revenue and liquid welfare (a welfare notion suitable for budgeted settings) respectively.
Consider now multi-unit (homogeneous item) settings and downward-sloping valuations. The \emph{uniform price} auction gives the $m$ available units to the bidders with the $m$ highest marginals, charging the $(m+1)$st highest marginal. Even with correlation, every BNE of this auction that satisfies a no-overbidding condition guarantees $1/4$ of the optimal expected welfare~\cite[]{KeijzerMST13,BabaioffLNL14}. \citet{DevanurMSW15} consider a different simple bidding format, where the bidders only submits one real number as their bid. They analyze a mechanism they name `Single-Bid', and show this mechanism gives $O(\log m)$ approximation for sub-additive valuations, which is essentially tight for the class of mechanisms they consider~\cite[]{BravermanMW16}. \citet{FeldmanFMR16} analyze this mechanism in settings with restricted complements.

\paragraph{Interdependent Values in AMD}
Interdependent values model has recently gained attention in the algorithmic mechanism design community. This interest was initiated by~\citet{RoughgardenT16} and~\citet{Li17} which successfully apply the simple vs. optimal paradigm to the interdependent setting. \citet{RoughgardenT16} show that a prior independent, single sample mechanism can approximate the optimal revenue under sufficient conditions, and very similar conditions are used by~\citet{Li17} to show that running VCG with monopoly reserves can approximate the optimal welfare as well. \citet{CFK} focused on removing assumptions. They showed a mechanism that approximates the optimal revenue for Matroid environments that does not make any distributional assumptions. The assumptions they make on the valuation functions are single-crossing, and the submodular condition we use for  our positive result in the multi-item setting. \citet{EdenFFG18} inspected what happens if we remove the single-crossing assumption in the context of welfare maximization. Without any assumptions on the valuations, they notice things can turn out arbitrarily bad. Therefore, they inspected the $c$-SC condition we inspect in our paper, and showed improved approximation bounds that depend on the parameter $c$, and sometimes depend on $n$ as well. \citet{EdenFFGK19} do not make any single-crossing-type assumption, and focus instead on a submodularity assumption over signals (which they call SoS). They notice that this condition alone suffices to give constant factor approximation for welfare and revenue in single-parameter domains, and also in general combinatorial settings when the affect of one's signal on her value is independent from the affect of other bidders' signals on one's value. The latter result holds even if signals are multidimensional.



\section{Preliminaries}
\label{sec:prelim}
\paragraph{Notation.}
Let ${\bf x}, {\bf y}$ be vectors, $i$ an index and $A$ a subset of indices. Then ${\bf x}_A$ is the vector obtained by taking indices $A$ of $x$. Let ${\bf x}_{-A}={\bf x}_{\overline{A}}$ 
(in particular ${\bf{x}}_{-i}$ is vector $\bf x$ with the $i$th entry removed). 
By ${\bf x} \le {\bf y}$ we mean $x_i \le y_i$ for all $i$.
Vectors $\bf{1}$ and $\bf{0}$ are the all-ones and all-zeros vectors.

\subsection{Basic Problem Setting}
The following is a standard interdependence setting for selling a single item \citep{MilgromW82}: There are $n$ bidders (agents), each bidder $i$ with a privately-known \emph{signal} $s_i$ from signal space $S_i$ (a continuous interval in $\mathbb{R}_{\ge 0}$). 
Let $\bs=(s_1,\dots,s_n)\in S\subseteq S_1\times\dots\times S_n$ denote the signal profile of all bidders. 
Every bidder $i$ also has a publicly-known \textit{valuation} $v_i:S\rightarrow\mathbb{R}_{\ge 0}$, which is a function of the signal profile~$\bs$. This dependence of the values on other bidders' signals is the defining property of interdependence. 
Function $v_i$ is weakly increasing in each coordinate and strictly increasing in $s_i$.
Let $\bv=(v_1,\dots,v_n)$ denote the valuation profile of all bidders.

We also consider combinatorial settings in which there are $m$ items for sale:
For multiple items, $\s_i$ is multi-dimensional and there is a signal $s_{i\ell}$ for every item~$\ell$. Value $v_{i\ell}$ of the $\ell$th item is a function of $\s_\ell=(s_{1\ell},\ldots,s_{n\ell})$. We focus on \emph{unit-demand} valuations, for which the value for a subset of items $T\subseteq [m]$ is $v_{i}(T;\s)=\max_{\ell\in T}\{v_{i\ell}(\s_\ell)\}$.


As part of the classic interdependence model, valuation profiles are assumed to be \emph{single-crossing (SC)} 
\cite[e.g.][]{MilgromW82,Maskin96,Ausubel00}.

\begin{definition}[SC]
	Valuation profile $\bv$ is \emph{SC} if for every bidder pair $i, i'$, every signal profile $\bs$ and every $\delta\ge 0$,
$$v_i(s_i+\delta,\snoi)-v_i(\bs)\ge v_{i'}(s_i+\delta,\snoi)-v_{i'}(\bs).
$$
\end{definition}
Intuitively SC means that a bidder's own value is influenced the most from a change in her signal.
\citet{EdenFFG18} introduced a natural relaxation called \emph{$c$-SC} for $c\ge 1$, 
which requires $c\left( v_i(s_i+\delta,\snoi)-v_i(\bs)\right) \ge v_{i'}(s_i+\delta,\snoi)-v_{i'}(\bs)$. Our results generalize to this relaxation. 

\paragraph{Who Knows What.}
We address both full-information and Bayesian settings. 
In either model the signal spaces $S_1,\dots,S_n$ and the valuation functions $v_1,\dots,v_n$ are assumed to be public knowledge.
In full information settings, all bidders know all signals (but the mechanism does not). 
In Bayesian settings (our main focus), bidder $i$ has private knowledge of $s_i$, which is drawn from a publicly-known distribution $F_i$ with density~$f_i$. We denote by $F$ the publicly-known joint distribution of signal profile~$\bs$.

\paragraph{Correlation.}
Even if the signals are sampled independently, the values are correlated since they depend on the same signals. 
This generalizes various types of correlation found in other papers~\cite[]{BateniDHS15,ChawlaMS15,ImmorlicaSW20}.
When dealing with a single item we further allow the signals themselves to be correlated, i.e., $F$ is not necessarily a product distribution. 
We denote by $F_{\mid s_i}$ the distribution of $\snoi$ given signal $s_i$. 

\paragraph{Submodularity.}
Submodularity of the valuation functions is assumed in Section \ref{sec:multiple-pos}; it means that as a bidder's own signal increases, an increase in the others' signals has diminishing marginal influence on her value:%
\footnote{\cite{EdenFFGK19} define a stronger version of valuation submodularity called submodularity over signals (SoS).} 

%


\begin{definition}[\cite{CFK}]\label{cond:weak_submod}
	A valuation $v_i$ is \emph{submodular}
	if for every signal $s_i$, every two profiles of the other bidders $\snoi\le {\bs'}_{-i}$ and every
	$\delta\ge 0$,
	$$v_i(s_i+\delta,\s_{-i}) -v_i(s_i,\s_{-i})\geq v_i(s_i+\delta,\s'_{-i}) -v_i(s_i,\s'_{-i}).$$
\end{definition}
For multiple items, the same definition applies to $v_{i\ell}$ (rather than $v_i$) with $s_i,\s_{-i}$ replaced by $s_{i\ell},(\s_{\ell})_{-i}$.

\subsection{Mechanisms}

\paragraph{Objective.}
Our goal is to maximize welfare, that is, for a single item to allocate it to the bidder with the highest value $v_i(\bs)$, and for multiple items and unit-demand bidders to find a matching of items to bidders with the maximum total value. We denote the optimal welfare for a given setting with signal profile $\bs$ by $\opt(\bs)$, and the expected optimal welfare by $\opt=\mathbb{E}_{s\sim F}[\opt(\s)]$. 

\paragraph{Simple Mechanisms.}
A mechanism consists of a pair $(x,p)$ of (deterministic) allocation rule~$x$ and payment rule $p$. 
The mechanism solicits a signal report (\emph{bid}) $b_i$ from each bidder $i$. Let $\bb=(b_1,\dots,b_n)$ denote the bid profile.
The mechanism outputs for every bidder $i$ an indicator $x_i(\bb)\in \{0,1\}$ of whether she wins the item, and her payment $p_i(\bb)$. The allocation rule $x$ guarantees feasibility, i.e., $\sum_i x_i(\bb)\le 1$ for every $\bb$. Bidder~$i$'s expected utility given bid profile $\bb$ and true signal profile $\bs$ is $u_i(\bb,\s) = x_i(\bb)v_i(\bs)-p_i(\bb)$. 

\paragraph{No-overbidding (NOB).}
As standard in the literature on PoA \cite[e.g.][]{BhawalkarR11} we often assume NOB, defined as follows: 

\begin{definition}[Single-item NOB for interdependence]
	Bid profile $\bb$ satisfies \emph{NOB} if $b_i \leq s_i$ for every bidder $i$. Strategy profile $\sigma$ satisfies \emph{NOB} if for all $\bb \sim \sigma(\s)$, $\bb$ satisfies NOB.
\end{definition}

NOB assumptions reflect bidders' reluctance to expose themselves to negative utility from overbidding and help explain the prevalence of certain auction formats like the second-price auction in practice.
It is well-known that without NOB the second-price auction has unbounded PoA.%
\footnote{Consider two bidders with (independent private) values $\epsilon,1$ where $\epsilon\ll 1$. Bidding $1,0$ is an equilibrium of the second-price auction and its welfare is $\epsilon$, but the optimal welfare is $1$. In this equilibrium, the first bidder is bidding much higher than her value.}

%
%
%

\paragraph{Standard Mechanisms.}
In mechanisms for IDV, the \textit{highest-valued bidder} is computed according to solicited bids $\bids$ and the public-knowledge valuation profile~$\bf v$.
 
\begin{definition}[Critical bid]
	Given a partial bid profile $\bb_{-i}$, bidder $i$'s critical bid $b_i^*$ is the lowest report for which she wins the item, i.e., $b_i^* = \min\{b_i \mid x_i(\bb)=1\}$.
\end{definition}
\begin{definition}
	\label{def:GVA}
	The \emph{generalized Vickrey auction} (GVA) solicits signal bids, allocates the item to the highest-valued bidder $i$, and charges her critical bid value $v_i(b_i^*,\bb_{-i})$.
\end{definition}
There are also natural generalization of the first- and second-price auctions to interdependence (throughout we refer to these and not the independent private values versions).

\begin{definition}[Second-price (2PA) with interdependence]  
	\label{def:second-price-auction}
	The second-price auction solicits signal bids, allocates the item to the highest-valued bidder, and charges the second-highest value as payment.
\end{definition} 

\begin{definition}[First-price with interdependence]  
	\label{def:first-price-auction}
	The first-price auction solicits signal bids, allocates the item to the highest-valued bidder, and charges her value as payment.
\end{definition}

\subsection{Price of Anarchy Background}
\label{sec:equilibrium-notions}
Let $\sigma_i(\cdot)$ be bidder $i$'s bidding \emph{strategy} as a function of her true signal.
A bidding strategy can be \emph{pure} (mapping to a single reported signal) or \emph{mixed} (mapping to a distribution over reported signals).
 
Fix a mechanism. 
Given a strategy profile $\sigma=(\sigma_1,\ldots,\sigma_n)$, 
let $\sigma(\s)$ be a 
mapping of signals $\s$ to bids, and $\sigma_{-i}(\s_{-i})$ be the mapping excluding agent $i$'s bid.
%
We focus on the following equilibrium concepts:
\begin{enumerate}
	\item Ex-post equilibrium (EPE): A bidding strategy profile $\sigma$ constitutes an EPE if $\sigma$ is deterministic,%
	\footnote{Ex-post equilibria are not necessarily deterministic, but as we use these equilibria in the context of lower bounds, the restriction to deterministic only strengthens our results.}
	and for every bidder $i$ with signal $s_i$, and every signal profile $\s_{-i}$, $u_i(\sigma(\s);\s)\geq u_i((b_i,\sigma(\s_{-i}));\s)$ for every $b_i$.
	\item Pure Nash equilibrium (PNE):
	A bidding strategy profile $\sigma$ constitutes a PNE of a full information setting with signal profile $\bs$ if $\sigma$ is deterministic, and for every bidder $i$,  $u_i(\sigma(\s);\s)\geq u_i((b_i,\sigma(\s_{-i}));\s)$ for every $b_i$.
	\item Bayes-Nash equilibrium (BNE):
	A (possibly randomized) bidding strategy profile $\sigma$ constitutes a BNE of a Bayesian setting with signal distribution $F$ if for every bidder $i$ and signal $s_i$, and every $b_i$, $$\E_{\snoi\sim F_{\mid s_i}}\left[ u_i(\sigma(\s);\s)\right] \geq \E_{\snoi\sim F_{\mid s_i}}\left[u_i((b_i,\sigma(\s_{-i}));\s)\right].$$
\end{enumerate}

\paragraph{Comparisons.} Note that the difference between PNE and EPE is that in a PNE, the bidder takes into account other bidders' signals, wheres, in an EPE, the bidder is not aware of other bidders' signals, and only knows they bid according to their bidding strategy. Moreover, although both BNE and EPE do not require the bidder to know other bidders' signals, EPE is a stronger notion since it requires the bidder to be best responding with respect to every realization of other bidders' signals, whereas BNE only requires the bidder to be best responding in expectation over other bidders' signals. 

The following observation summarizes the relation between the concepts.
\begin{observation}[Equilibrium hierarchy]
	\label{obs:equil-hierarchy}
	Fix a mechanism and a signal profile $\bs$. Every $\epe$ is also a $\pne$ w.r.t. $\s$, and every $\pne$ is a $\bne$ for an appropriate Bayesian setting. 
\end{observation}

The proof of the above observation is deferred to Appendix~\ref{appx:poa}.

\paragraph{Price of Anarchy (PoA).} 
The expected welfare achieved by a single-item mechanism $(x,p)$ for signal profile $\bs$ and bid profile $\bb$ is $\sum_{i} x_i(\bb)v_i(\bs)$. Where the mechanism (in particular, the allocation rule) is evident from the context, we denote this by $\SW(\bids,\bs)$. For a 
equilibrium strategy $\sigma$, let $\eq(\sigma,\s)=\E[\SW(\sigma(s),\s)]$ be the expected welfare of $\sigma$ 
at 
$\s$, and $\eq(\sigma)=\E_{s\sim F}[\eq(\sigma,\s)]$. The ex-post PoA, Nash PoA and Bayesian PoA of an auction, respectively, are:
$$
\epoa=\sup\limits_{\substack{\bs,\bv,\sigma:\\\sigma\text{ is EPE}}}\frac{\opt(\bs)}{\eq(\sigma,\bs)};~~~
\npoa=\sup_{\substack{\bs,\bv,\sigma:\\\sigma\text{ is PNE}}}\frac{\opt(\bs)}{\eq(\sigma,\bs)};~~~
\bpoa=\sup\limits_{\substack{F,\bv,\sigma:\\\sigma\text{ is BNE}}}\frac{\opt}{\eq(\sigma)}.
$$

In the above, the supremum is taken over \emph{all} 
Bayesian settings. It is also possible to restrict the class of settings, for example, to those with valuation profiles that satisfy the SC property, or bidding strategies that satisfy a NOB property. Such restriction will, in general, improve the PoA, and we leverage this method to achieve our results. 

Observation~\ref{obs:equil-hierarchy} implies that lower bounds for $\epoa$ and upper bounds for $\bpoa$ propagate to the other notions.

\section{Single Item: A Positive Result}
\label{sec:single}
%
In this section, we focus on single-item settings. Our main result is a (parameterized) property which, along with SC, leads to good PoA bounds.

Recall the setting described above, where the signal of, say, agent 3 affects only the value of herself and agent 1, and no other agent.
In the opposite extreme, a change in an agent's signal affects the values of all other agents equally. That is: 
$$
\forall i,j,j' : v_j(s_i+\delta, {\bf s}_{-i}) - v_j({\bf s}) = v_{j'}(s_i+\delta, {\bf s}_{-i}) - v_{j'}({\bf s}).
$$
We call this condition homogeneous influence. One can verify that homogeneous influence holds in our running example -- the resale model, as well as in other prominent settings like Klemperer's wallet-game \cite[]{klemperer1998auctions}, common-values \cite[]{wilson1969communications}, private values, and private/common value interpolation~\cite[]{bergemann2013robust}.
We find that homogeneous influence, together with SC, ensures full efficiency (this is a direct corollary of Corollary~\ref{prop:GVA:gamma-bound-sc:gamma-BPOA}).
\begin{proposition}
	In single item settings with SC and homogeneous influence, every equilibrium of the GVA is fully efficient.
\end{proposition}


The following question arises: how robust is this result? In other words, how would the PoA deteriorates as we move away from homogeneous influence? To this end, we introduce the following parameterized property, which measures how far a valuation profile is from homogeneous influence.

\begin{definition}[$\gamma$-heterogeneity in signal-value impact] \label{def:influence}
	A valuation profile is $\gamma$-heterogeneous in signal-value impact, or in short \emph{$\gamma$-heterogeneous}, if for every agent $i$, two other agents $j,j'$, signal profile ${\bf s}$ and $\delta>0$, 
	$$\gamma (v_j(s_i+\delta, {\bf s}_{-i}) - v_j({\bf s})) \geq v_{j'}(s_i+\delta, {\bf s}_{-i}) - v_{j'}({\bf s}).
	$$
\end{definition}

By definition, $\gamma$ is always at least $1$; the special case of $\gamma=1$ is homogeneous influence. 
Our main result in this section is a tight bound for the $\bpoa$ for $\gamma$-heterogeneous, SC valuations. Our result extends beyond SC, to $c$-SC profiles \cite[]{EdenFFGK19}, where the PoA degrades gracefully with the parameter $c$ (SC is $c$-SC for $c=1$).

%

\begin{theorem}[Main positive result for single-item]
	\label{prop:upper-bound:gamma-bound+c-sc:GVA}
	Consider a single-item setting with $\gamma$-heterogeneous, $c$-SC, continuous valuations.	
	The $\bpoa$ of GVA under NOB is bounded by $1+\max\{\gamma,c\}$. 
	For SC valuations the $\bpoa$ is tightly $\gamma$.
	These results hold even for correlated signals.
\end{theorem}

Theorem~\ref{prop:upper-bound:gamma-bound+c-sc:GVA} is essentially tight, as illustrated in the following theorem.
\begin{restatable}{theorem}{singlelb}
	\label{pro:gamma-bound-csc:POA-close-gamma-c}
	There exist single-item settings satisfying $\gamma$-heterogeneity and $c$-SC  such that the $\npoa$ of the generalized Vickrey auction is arbitrarily close to $\max\{c,\gamma\}$, even under no-overbidding.\footnote{The theorem holds even under the additional assumption of Submodularity over Signals, for more information refer to \cite{EdenFFGK19}}
\end{restatable}

The proof of this lower bound is deferred to the end of this section.
We remark that the upper bound of $1+\max\{\gamma,c\}$ and the lower bound apply also with respect to 2PA with interdependence, see Appendix~\ref{appx:2PA} for more details.

Before proving Therom~\ref{prop:upper-bound:gamma-bound+c-sc:GVA}, we introduce following technical lemma and corollary which will be useful throughout the paper.

\begin{restatable}{lemma}{techlemma}
	\label{lem:j-bound-i-diff}
	Let $\mathbf{v}$ be a valuation profile satisfying $\gamma$-heterogeneity and $c$-SC.
	Then for every two agents $i,j$, signal profile ${\bf s}$ and coordinate-wise non-negative vector $\delta = (\delta_1,\ldots,\delta_{i-1},0,\delta_{i+1},\ldots,\delta_n)$, 
	$$
	v_j(\s+\delta)-v_j(\s) \geq \Big(v_i(\s+\delta)-v_i(\s)\Big)/\max\{\gamma,c\}.
	$$
\end{restatable} 

\begin{proof}
	Assume by renaming that $i=1$ and $j=n$; therefore, $\delta_1=0$. Let $\s^\ell=(s_1+\delta_1, s_2+\delta_2,\ldots,s_\ell+\delta_\ell, s_{\ell+1},\ldots, s_n);$ that is, the vector $\s$ after increasing the first $\ell$ coordinates by their respective $\delta$'s. Note that $\s^1=\s$. We have that 
	\begin{eqnarray*}
	v_j({\s} + \delta) -v_j(\s)  & =& v_n({\s} + \delta) -v_n(\s)\\ 
	& = &v_n(\s^n) - v_n(\s^{n-1})+\sum_{k=2}^{n-1}\left(v_n(\s^k)-v_n(\s^{k-1})\right)\\
	& \geq &\left(v_1(\s^n) - v_1(\s^{n-1})\right)/c + \sum_{k=2}^{n-1}\left(v_1(\s^k)-v_1(\s^{k-1})\right)/\gamma\\
	& \geq &\left(v_1(\s^n) - v_1(\s^{n-1})\right)/\max\{\gamma,c\} + \sum_{k=2}^{n-1}\left(v_1(\s^k)-v_1(\s^{k-1})\right)/\max\{\gamma,c\} \\
	& = &(v_1({\s} + \delta) -v_1(\s))/\max\{\gamma,c\} \\ & = & (v_i({\s} + \delta) -v_i(\s))/\max\{\gamma,c\},
	\end{eqnarray*}
	where the first inequality follows the $c$-SC and $\gamma$-heterogeneous assumptions.
\end{proof}

The following is a direct corollary of Lemma~\ref{lem:j-bound-i-diff}.
\begin{corollary}
	\label{col:j-bound-i}
	Let $\mathbf{v}$ be a valuation profile satisfying $\gamma$-heterogeneity and $c$-SC.
	For every two agents $i,j$, signal profile ${\bf s}$ and coordinate-wise non-negative vector $\delta = (\delta_1,\ldots,\delta_{i-1},0,\delta_{i+1},\ldots,\delta_n)$, if $v_j(\s) \geq v_i(\s)/d$ for $d\geq \max\{\gamma,c\}$, then $v_j(\s+\delta) \geq v_i(\s+\delta)/d.$
\end{corollary}

\begin{proof}
	By Lemma~\ref{lem:j-bound-i-diff},
	$v_j(\s+\delta)-v_j(\s) \geq \left(v_i(\s+\delta)-v_i(\s)\right)/\max\{\gamma,c\}\geq \left(v_i(\s+\delta)-v_i(\s)\right)/d.$
	Combining this with the fact that $v_j(\s) \geq v_i(\s)/d$ establishes the Lemma.
\end{proof}

We now proceed with the proof of \cref{prop:upper-bound:gamma-bound+c-sc:GVA}.

\begin{proof}[Proof of \cref{prop:upper-bound:gamma-bound+c-sc:GVA}]
	Let $\bs$ be a signal profile and $\bids$ be a bid profile such that $\bids \leq \s$. Let $i\in \argmax_j v_j(\s)$, and $w(\bids)$ be the winner under $\bids$. We prove that
	\begin{eqnarray}
	u_i((s_i,\bids_{-i});\s) \geq v_i(\s) - \max\{\gamma,c\}v_{w(\bids)}(\s).\label{eq:single_item_smoothness}
	\end{eqnarray}
	We distinguish between the following two cases: 
	
	 {\bf Case 1:} Bidder $i$ does not win the item under bidding profile $(s_i,\bids_{-i})$. Thus, her utility is $0$. Moreover, there exists a bidder $j\neq i$ such that $v_j(s_i,\bids_{-i})\geq v_i(s_i,\bids_{-i}).$ 
	Since bidder $w(\bids)$ wins at $\bids$, $v_{w(\bids)}(\bids)\geq v_j(\bids).$
	By Lemma~\ref{col:j-bound-i} (using $d=\max\{\gamma,c\}$),
	$$\max\{\gamma, c\} v_{w(\bids)}(s_i,\bids_{-i})\ \geq\  v_j(s_i,\bids_{-i})\ \geq\ v_i(s_i,\bids_{-i}).$$ Applying Lemma~\ref{col:j-bound-i} again, we get $$\max\{\gamma, c\} v_{w(\bids)}(\s) \geq v_i(\s).$$ 
	Therefore,
	$$u_i((s_i,\bids_{-i}); \s) \ =\ 0 \ \geq\ v_i(\s)-\max\{\gamma,c\}v_{w(\bids)}(\s).$$
	
	 {\bf Case 2:} Bidder $i$ wins the item under bidding profile $(s_i,\bids_{-i})$. Let $b^*_i\leq s_i$ be $i$'s critical bid given $\bids_{-i}$.  $u_i((s_i,\bids_{-i});\s)=v_i(\s)-v_i(b^*_i,\bids_{-i})$.  If $w(\bids)=i$, $$u_i((s_i,\bids_{-i});\s) \geq 0 \geq   v_i(\s)- \max\{\gamma, c\} v_{w(\bids)}(\s).$$ 
	Otherwise, $w(\bids) \neq i$.
	Let $b'_i \geq b_i$ be the lowest signal such that $v_i(b'_i, \bids_{-i}) \in \argmax_k{v_k(b'_i, \bids_{-i})}$. There exists such $b'_i \leq s_i$ since $i$ does not win under bidding profile $\bids$ but wins under bidding profile $(s_i, \bids_{-i})$. By the same argument and the continuity of the valuations, there exists $j\neq i$ such that $v_j(b'_i, \bids_{-i}) \in \argmax_k{v_k(b'_i, \bids_{-i})}$.
	By Corollary~\ref{col:j-bound-i},  $$\max\{\gamma, c\} v_{w(\bids)}(b'_i,\bids_{-i})\geq v_j(b'_i,\bids_{-i})=v_i(b'_i,\bids_{-i})\geq v_i(b^*_i,\bids_{-i}),$$ where the last inequality is due to the critical bid definition. Therefore,
	$$u_i((s_i,\bids_{-i});\s) \geq v_i(\s)-\max\{\gamma, c\} v_{w(\bids)}(b'_i,\bids_{-i})
	\geq v_i(\s) - \max\{\gamma, c\} v_{w(\bids)}(\s),$$
	where the last inequality follows by $\bids \leq \s$ and the monotonicity of the valuations. 
	This concludes the proof of \eqref{eq:single_item_smoothness}
	
	We can now establish the bound on the $\bpoa$:
	Let $\sigma=(\sigma_1,\ldots, \sigma_n)$ be a BNE satisfying NOB; that is, $\bids\leq \s$ for any $\bids$ in the support of $\sigma(\s)$. 
	Let $\mathbb{I}_i(\s)$ be the indicator variable of the event $i=\argmax_j v_j(\s)$ (breaking ties arbitrarily). Recall that $\eq(\sigma)$ denotes the welfare in equilibrium $\sigma$. 
	 We get: 
	\begin{eqnarray}
	\eq(\sigma) & \geq& \E_{\s}\Big[\sum_i u_i(\sigma(\s);\s)\Big] \label{eq:11}\\
	&\geq &\sum_i \E_{\s}[u_i((s_i,\sigma_{-i}(\s_{-i}));\s)] \label{eq:22}\\
	& \geq &\sum_i \E_{\s}[\mathbb{I}_i(\s)\cdot(v_i(\s)-\max\{\gamma,c\}v_{w(\sigma(\s))}(\s))]\label{eq:33}\\
	& = &\E_{\s}\Big[\sum_i \mathbb{I}_i(\s)v_i(\s)\Big] - \max\{\gamma,c\}\E_{\s}\Big[\sum_i \mathbb{I}_i(\s)v_{w(\sigma(\s))}(\s)\Big] \nonumber\\
	& = &\E_{\s}[\max_i v_i(\s)]-\max\{\gamma,c\}\E_{\s}[v_{w(\sigma(\s))}(\s)] \nonumber\\
	& = &\opt\ -\  \max\{\gamma,c\}\eq(\sigma)\nonumber,
	\end{eqnarray} 

	where \eqref{eq:11} holds since the sum of utilities is dominated by the welfare. 
	\eqref{eq:22} holds by the equilibrium hypothesis (and linearity of expectation).
	\eqref{eq:33} holds by \eqref{eq:single_item_smoothness}.
	With exact SC (i.e., $c=1$), the bound further improves to $\gamma$ (see \cref{prop:GVA:gamma-bound-sc:gamma-BPOA}).
\end{proof}

\begin{proof}[Proof of Theorem~\ref{pro:gamma-bound-csc:POA-close-gamma-c}]
		\par{\bf Case 1: $\npoa$ is $c$.}
		Consider the signal spaces $S_1 = S_2 = S_3 =[0,1]$ and the following valuations:
		$$
		\begin{array}{ll}
		v_1 \ =\ & \beta s_1 + 1\\
		v_2 \ =\ & c\beta s_1 + \epsilon s_2\\
		v_3 \ =\ & c\beta s_1 + \epsilon s_3
		\end{array}
		$$
		where $\epsilon$ is close to 0.

		Let the real signal be ${\bf s} ={\bf 1}$.
		We claim that bid vector $(0,1,1)$ is an equilibrium; agent $1$ wins, and pays either $1$ in generalized Vickrey (critical bid payment), and she can not increase her utility by increasing her bid (her payment is minimal with respect to $b_{-1}$).
		The other agents can not increase their bids by the no overbidding condition, and can not win by decreasing their bids, thus, both of the agents do not have a beneficial deviation.
		
		In this equilibrium, the achieved welfare is $\beta + 1$, but the maximum welfare is $c\beta + \epsilon$. As $\beta$ can be arbitrarily large, the $\npoa$ is arbitrarily close to $c$.
		
		\par{\bf Case 2: $\npoa$ is $\gamma$.}
		Consider the signal spaces $S_1 = S_2 = S_3 =[0,1]$ and the following valuations:
		$$
		\begin{array}{ll}
		v_1\ =\ & \gamma\beta s_1\\
		v_2\ =\ & \beta s_1 + s_2\\
		v_3\ =\ & \gamma\beta s_1 + \epsilon s_3
		\end{array}
		$$
		where $\epsilon$ is close to 0.
		
		Let the real signal be ${\bf s} ={\bf 1}$.
		We claim that bid vector $(0,1,1)$ is an equilibrium; agent $2$ wins and pays $\epsilon$, and in she can not increase her utility by increasing her bid (her payment is minimal with respect to $b_{-2}$).
		Agent $1$ can not win by increasing her bid, because for every bid she send, agent $3$ has a larger value than her, and she can not decrease her bid due to its minimality in her signal space; moreover, agent $3$ can not increase her bid due to its maximality in her signal space, and can not win the item by decreasing her bid.
		Thus, both agents do not have a beneficial deviation.
		
		In this equilibrium, the achieved welfare is $\beta + 1$, but the maximum welfare is $\gamma\beta + \epsilon$. As $\beta$ can be arbitrarily large, the $\npoa$ is arbitrarily close to $\gamma$.
\end{proof}

\section{Multiple Items: A Positive Result}
\label{sec:multiple-pos}
In this section, we study settings with multiple items and unit-demand agents. After some multi-item preliminaries, we identify a condition limiting knowledge asymmetry among agents, which can be used to accommodate positive PoA results. We then use it to prove a PoA guarantee for a simple simultaneous auction in the $n\ge m$ regime (\cref{thm:multi-pos}). (For the $n \ll m$ regime see Section~\ref{sec:multiple-neg}.)
\subsection{Multi-item Preliminaries}
\label{sub:multiple-prelim}




\paragraph{Participation.}
Similarly to \cite{SyrgkanisT13}, we assume agents can communicate to the item auctions whether or not they wish to participate and compete for the item being sold (in addition to their signal report).
The overall report of agent~$i$ is $(b_i,a_i)$, where $b_i$ is the agent's vector of reports (bids) of her $m$ signals, and ${a}_i$ is a participation vector with $a_{i\ell}=1$ if the agent participates in the auction for item $\ell$ and $a_{i\ell}=0$ otherwise. 
The agent's (mixed) strategy $\sigma_i(s_i)$ given her vector of signals is then a distribution over possible reports $(b_i,a_i)$.
In Appendix~\ref{appx:multiple-pos} we show that if agents are required to participate in {\em all} auctions, this leads to a bad PoA.

%
%

\paragraph{Notation and Assumptions.}
Let $\bids$ (respectively, $\avec$) be a profile of $n$ signal reports $b_i$ (respectively, $a_i$), and denote by
$\sigma(\s)$ a strategy profile given signal profile $\s$, that is, a distribution over $(\bids,\avec)$ pairs.
We use $X_i(\bids,\avec)$ to denote the item subset allocated in total by the separate auctions to agent~$i$, given the bid and participation profiles $\bids,\avec$. 

In this section, we state our results for valuation profiles satisfying c-SC and $\gamma$-heterogeneity.
Standard NOB assumptions for multiple items essentially restrict the sum of every agent's bids for items in a subset $X$ to at most the agent's value for $X$. 
We now formalize this notion under IDV:

\begin{definition}[Multi-item NOB for IDV]
	\label{def:multi-nob}
	With multiple items, a strategy profile $\sigma=(\sigma_1,\dots,\sigma_n)$ satisfies NOB if for every agent $i$ with signal $s_i$ it holds that 
	
	\begin{eqnarray}
	\E_{\substack{\s_{-i}\sim F_{-i},\\(\bids,\avec)\sim \sigma(\s)}} \Big[\sum_{\ell \in X_i(\bids,\avec)}v_{i\ell}(b_{i\ell},\s_{-i\ell})\Big] 
	\leq \E_{\substack{\s_{-i}\sim F_{-i},\\(\bids,\avec)\sim \sigma(\s)}} \left[v_{i}(X_i(\bids,\avec);\s)\right].
	\end{eqnarray}

\end{definition}
Note that this coincides with the standard NOB assumption if valuations are private.
\subsection{The Condition of Limited Knowledge Asymmetry}\; 
\label{sub:sufficient-cond}

We formulate a condition 
that enables our positive PoA result.
We use the following definition:

\begin{definition}[Truncated values and welfare] 
	\label{def:truncated}
	For every agent $i$ and item $\ell$, 
	the \emph{truncated value} $\tilde{v}_{i\ell}$ given signal profile $\s$ is 
	$\tilde{v}_{i\ell}(\s)=\tilde{v}_{i\ell}(\s_\ell)=\min_{j\neq i}v_{i\ell}(\s_{-j\ell},0_{j\ell}).$
 	The \emph{truncated optimal matching} $\widetilde{m}(\s)$ is a matching in $\argmax_{\text{matching }\mu}\{\sum_{(i,\ell)\in \mu}\tilde{v}_{i\ell}(\s)\}$, (breaking ties arbitrarily),
 	and the \emph{expected truncated welfare} is $\tildeopt=\E_{\s\sim F}\big[\sum_{(i,\ell)\in\widetilde{m}(\s)}\tilde{v}_{i\ell}(\s)\big]\nonumber\nonumber$.
\end{definition}
In words, truncated value $\tilde{v}_{i\ell}(\s)$ is agent $i$'s value for item $\ell$ after the most significant signal except for $i$'s own has been zeroed out; $\tildem(\s)$ is an allocation that maximizes the social welfare with respect to the truncated values, given~$\s$; and $\tildeopt$ is the optimal social welfare in expectation over $\s$ with respect to truncated values.
Observe that by the monotonicity of values in signals, $\tildeopt\le \opt$. 

Given an interdependent setting, the \emph{limited knowledge asymmetry} condition is the following: there exists a constant $d$ such that 

 \begin{equation}
d\cdot\tildeopt \geq \opt.\label{eq:condition}
\end{equation}

Intuitively, if a truncated value $\tilde{v}_{i\ell}(\s)$ is far from the true value $v_{i\ell}(\s)$, this means that agent $i$'s value for item $\ell$ is largely determined by the information that is held by a single agent who is not $i$ herself. 
This means there is information asymmetry among $i$ and the informed agent.
The condition in \eqref{eq:condition} rules out settings in which nearly all the welfare stems from such information asymmetry. In extreme such settings, all information shaping the values can be traced to a single agent -- as in the classic and well-studied ``drainage tract'' model of Wilson \cite[]{wilson1969communications,milgrom2004putting}. We now give two opposite examples: a natural setting in which the condition in \eqref{eq:condition} holds (\cref{ex:weighted-sum-revisited}), and an extreme setting in which one agent has all the information (\cref{ex:tildeopt-is-far}). For the second example we show a lower bound of $\Omega(m)$ on the PoA of a natural family of mechanisms. 
\begin{example}[Weighted-sum valuations revisited]
	\label{ex:weighted-sum-revisited}
	Recall the Resale model, introduced in \cref{example:running}. We extend this example to multiple items: Assume for simplicity that $n=m$. For every unit-demand agent $i$ and item $\ell$, let $v_{i\ell}(\s)=s_{i\ell} + \beta\sum_{j\ne i} s_{j\ell}$ for $\beta\le 1$ (where the restriction on $\beta$ maintains the SC property). 
	We assume all signals are drawn i.i.d.~from a monotone hazard rate (MHR) distribution (such as a uniform, normal, or exponential).
\end{example}
\begin{restatable}{proposition}{weightedsumrevisited}
	\label{pro:weighted-sum-revisited}
	In \cref{ex:weighted-sum-revisited}, for a sufficiently large $n$, $(1+e)\tildeopt\ge \opt$.
\end{restatable}
The proof of \cref{pro:weighted-sum-revisited} appears in \cref{appx:multiple-pos}.
\begin{example}[Agent 1 holds almost all information]
	\label{ex:tildeopt-is-far}
	Consider $n$ unit-demand agents with single-dimensional signals $s_1,\dots,s_n$, and signal spaces $S_1 = [0,1]$ for agent 1 and $S_i = \{0\}$ for every agent $i\neq 1$. The agents have the following values for every item $\ell\in[n]$: Agent 1's value is
	$v_{1\ell} =  s_{1}$, and for every agent $i \neq 1$, $v_{i\ell} = (1-\epsilon)s_{1}$. 
	Note these valuations are SC and $\gamma$-homogeneous.
\end{example}

In Example~\ref{ex:tildeopt-is-far}, the truncated values are $\tilde{v}_{1\ell} = s_1$ and $\tilde{v}_{i\ell}=0$ for $i \neq 1$. We get that $\tildeopt = s_1$ while $\opt > n(1-\epsilon)s_1$, so there is no constant $d$ for which \eqref{eq:condition} holds.



\begin{restatable}{proposition}{multioptfartildeopt}
	\label{prop:multi-neg-for-2nd-price}
	For each bidder let $h_{i\ell}(\s)$ be a monotone function such that $\forall \s, h_{i\ell}(\s) \leq v_{i\ell}(\s)$, and $h_{i\ell}(s_{i\ell}, {\bf 0}_{-i\ell}) = v_{i\ell}(s_{i\ell}, {\bf 0}_{-i\ell})$.
	Every mechanism that allocates each item separately to the highest-valued agent according to $h_{i\ell}(\s)$, and charges an agent zero payment if no-one else participates, has $\epoa$ of~$\Omega(m)$ even under NOB.
	In particular, by $h_{i\ell}(\s) = v_{i\ell}(\s)$ we get that this holds for simultaneous 2PA or GVA auctions. 
\end{restatable}

%



The proof of \cref{prop:multi-neg-for-2nd-price} appears in \cref{appx:multiple-pos} and uses Example~\ref{ex:tildeopt-is-far}.  As shown in the next section, the above proposition holds for the mechanism we develop and use to get good PoA guarantees. 
We leave open the question of whether there exists a simple mechanism with good PoA for settings such as \cref{ex:tildeopt-is-far} that do not satisfy \eqref{eq:condition}.



\subsection{Positive Result: $n \geq m$ Regime}
We present a simple simultaneous item-bidding mechanism and show that every BNE of it achieves an $O(\max\{\gamma,c\}^2)$-approximation to $\widetilde{\opt}$. 
For valuations satisfying Equation~\eqref{eq:condition}, it implies an $O(d\max\{\gamma,c\}^2)$-approximation to $\opt$.
In the special case of SC (i.e. where $c=1$), the results suggests that every BNE achieves an $O(\gamma^2)$-approximation to $\widetilde{\opt}$ and for valuations satisfying Equation~\eqref{eq:condition}, it implies an $O(d\gamma^2)$-approximation to $\opt$.

\paragraph{Intuition.}
Before introducing the mechanism, we provide some intuition that helps shed light on the choice of our mechanism and its analysis.
Specifically, we emphasize how our analysis is different from the smoothness framework, which has become the standard technique for establishing PoA results for simple auctions.

A key step in the smoothness paradigm is to find an appropriate hypothetical deviation for each player such that for any bids of the other bidders, the utility achieved by the deviation is lower bounded by some fraction of the player's contribution to the optimal welfare, less some error term. 
For the case of independent private values and unit-demand bidders\footnote{For simultaneous second price auctions.},~\citet{ChristodoulouKST16} show that the hypothetical deviation of going {\em all-in} for item $\ell$ gives the following  guarantee:\footnote{Bidding her true value $v_{i\ell}$ for item $\ell$ and $0$ for all other items.}
\begin{equation}
	u_i((v_{i\ell},\mathbf{0}_{i-\ell}),\bids_{-i})\geq v_{i\ell} - \max_{j\neq i} b_{j\ell},
	\label{eq:smooth}
\end{equation}
Inequality \eqref{eq:smooth} suffices to give PoA bounds in full information settings. 
Moreover, this bound also applies to settings with incomplete information via an extension theorem, as described next. 

Even though the agent does not know the realization of other agents' valuations, she can sample from the value distribution of others, and go ``all-in'' for the item $\ell$ she receives in the optimal allocation in the sampled market, thus effectively simulating the contribution of this agent to the optimal assignment (the first term in the right-hand side of \eqref{eq:smooth}). This is termed the `doppelg{\"a}nger' technique. The error term is then bounded by tying it to an allocation for sampled equilibrium bids and using a NOB assumption.

Unfortunately, there are some major obstacles to generalizing this type of analysis to IDV. Whereas one can derive a smoothness-type inequality, its right-hand side expression would take the form of $\gamma\tilde{v}_{i\ell}(\s)-\max_{j\neq i}v_{j\ell}(b_{j\ell},s_{-j\ell})$. Using this inequality, one can derive PoA results for full information games\footnote{We do not include this derivation in the manuscript since it does not help us in proving a bound on the B-PoA.}, but these bounds would not generalize to incomplete information settings for the following reasons:
\begin{enumerate}[(a)]
	\item  The doppelg{\"a}nger technique cannot be directly applied. Agent $i$ is only aware of $s_i$ and is uninformed of 
	$\s_{-i}$, which affects her valuations. Therefore, sampling signals of other bidders $\bt_{-i}$, and going all-in accordingly, will simulate a different value distribution than the desired one, as the bid takes into account only $s_i$ and not $\s_{-i}$.
	\item The error term in \citet{ChristodoulouKST16} crucially depends only on the bids of other bidders, but in IDV the error term depends also on the agents' \textit{real} signals, which prevents from their analysis to go through. 
	\item  The benchmark in our setting should be the truncated optimal allocation and not the true optimal allocation.
\end{enumerate}

The following observation drives our analysis: 
since the agents have access only to their own signal, 
the doppelg{\"a}nger technique can be successfully applied when considering what we refer to as {\em privatized valuations}; namely, the valuation of agent $i$ under $i$'s private signal, while zeroing out the contribution of others' signals. We can then use the $\gamma$-heterogeneous condition to claim that the effect of other bidders' signals on the valuation of two different agents is roughly the same. 

The discussion above motivates the design of the mechanism we propose, which is a simultaneous second-price item auction with respect to the privatized valuations (see \cref{alg:multiple-pos}).
For the sake of analysis, we split the benchmark into two terms. 
The first term ($\mathsf{SELF}$) accounts for the privatized valuations, and the second term ($\mathsf{OTHER}$) accounts for the contribution of the signals of other bidders to the benchmark (see Equation~\eqref{eq:welfare-decomposition}). We proceed by providing smoothness-type inequalities for each of these terms separately (see Lemma~\ref{lem:multiple:bounded-utility}). Specifically, for each term, we identify an appropriate hypothetical deviation, such that the  expected welfare in equilibrium ``covers'' it. 
We proceed by formalizing the intuition given above.

{\it Privatized values.}\; Given a (reported) signal profile $\bb_\ell$ for item $\ell$, agent $j$'s \emph{privatized value} $\hat{v}_{j\ell}$ for $\ell$ is her value when other agents' signals are set to zero: $\hat{v}_{j\ell}(\bb_{\ell})=v_{j\ell}(b_{j\ell},\mathbf{0}_{-j\ell})$. We use $\hat{v}_{j\ell}(\bb_{\ell})$ and $v_{j\ell}(b_{j\ell},\mathbf{0}_{-j\ell})$ interchangeably.
Observe that the privatized value is upper-bounded by the truncated value $\hat{v}_{j\ell}(\bb_{\ell})\le \tilde{v}_{j\ell}(\bb_{\ell})$. 
The mechanism analyzed is this section, is essentially a second price auction on the privatized values of the agent.

\paragraph{Theorem Statement.} We focus on the simple mechanism described in \cref{alg:multiple-pos}, which consists of simultaneous \emph{privatized} second-price auctions. This mechanism allocates every item separately by running a second-price auction over the privatized values of agents participating for this item. Our main theorem in this section is the following:
\begin{figure}[H]
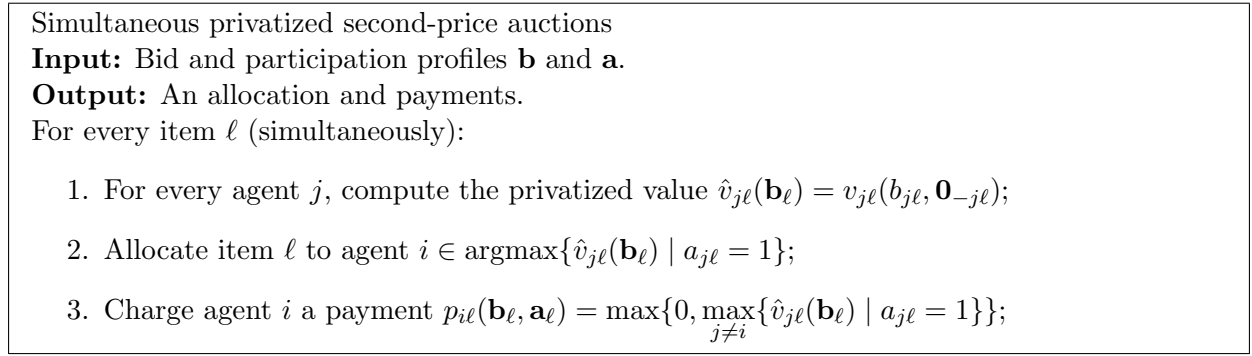

	\MyFrame{
		$\textrm{Simultaneous privatized second-price auctions}$\\
		\textbf{Input:} Bid and participation profiles $\bids$ and $\avec$.\\
		\textbf{Output:} An allocation and payments.
		
		 For every item $\ell$ (simultaneously):
		\begin{enumerate}
			\item For every agent $j$, compute the privatized value 
			$\hat{v}_{j\ell}(\bb_\ell)=v_{j\ell}(b_{j\ell},\mathbf{0}_{-j\ell})$;
			\item Allocate item $\ell$ to agent $i\in\argmax \{\hat{v}_{j\ell}(\bb_\ell)\mid a_{j\ell}=1\}$;
			\item Charge agent $i$ a payment $p_{i\ell}(\bb_\ell,\avec_\ell)=\max\{0,\max\limits_{j\neq i} \{\hat{v}_{j\ell}(\bb_\ell)\mid a_{j\ell}=1\}\}$;
		\end{enumerate}
	}
	\caption{Simultaneous privatized second-price auctions.}
	\label{alg:multiple-pos}
\end{figure}

	
	


		
		
		
		

\begin{restatable}{theorem}{multipleub}\label{thm:multi-pos}
	Consider a setting with $m$ items, $n \geq m$ unit-demand agents, and $\gamma$-heterogeneous, c-SC valuations. If $d\cdot\tildeopt\ge \opt$ for some constant $d$, then the $\bpoa$ of simultaneous privatized second-price auctions (See~\cref{alg:multiple-pos}) under no-overbidding is $O(d\max\{\gamma,c\}^2)$.
\end{restatable}

%

The condition that $d\cdot\tildeopt\ge \opt$ in \cref{thm:multi-pos} is necessary, as by $h_{i\ell}(\s) = v_{i\ell}(\s_{i\ell}, {\bf 0_{-i\ell}})$, Proposition~\ref{prop:multi-neg-for-2nd-price} holds for the simultaneous privatized second-price auction as well.
%
The remainder of the section is dedicated to proving \cref{thm:multi-pos}.
We use the following notation: Let $u_i((\bids,\avec);\s)$ denote agent $i$'s utility under bidding profile $(\bids,\avec)$ and signal profile $\s$. 
Let $p_{i\ell}(\bb_\ell,\avec_\ell)$ be as defined in \cref{alg:multiple-pos}.

We first decompose the welfare as stated above. Recall from \cref{def:truncated} that $\widetilde{m}(\s)$ is the matching that maximizes the social welfare with respect to the truncated values at signal profile $\s$. We decompose $\tildeopt$ as follows: 
\begin{eqnarray}
\tildeopt &=&  \E_{\s}\left[\sum_{(i,\ell)\in\widetilde{m}(\s)}\tilde{v}_{i\ell}(\s)\right]\nonumber\\& =& \underbrace{\E_{\s}\left[\sum_{(i,\ell)\in\widetilde{m}(\s)}v_{i\ell}(s_{i\ell},\mathbf{0}_{-i\ell})\right]}_{\text{$\mathsf{SELF}$}}\ +\  \underbrace{\E_{\s}\left[\sum_{(i,\ell)\in\widetilde{m}(\s)}\tilde{v}_{i\ell}(\s)-v_{i\ell}(s_{i\ell},\mathbf{0}_{-i\ell})\right]}_{\text{$\mathsf{OTHER}$}}.\label{eq:welfare-decomposition}
\end{eqnarray}

In order to prove \cref{thm:multi-pos}, we show that the equilibrium `covers' the two terms in the decomposition above. 
This is shown in Lemmas~\ref{lem:bounding-self} and~\ref{lem:bounding-other}, which are proved in Sections~\ref{sub:multi-pos-step2} and \ref{sub:multi-pos-step3}, respectively.

\begin{lemma}
	\label{lem:bounding-self}
	For every BNE $\sigma$, $2\eq(\sigma)\geq\mathsf{SELF}$.
\end{lemma}

\begin{lemma}
	\label{lem:bounding-other}
	For every BNE $\sigma$, 
	$O(\max\{\gamma,c\}^2)\cdot \eq(\sigma) \geq \mathsf{OTHER}.$ 
\end{lemma}

With these Lemmas in hand, we prove \cref{thm:multi-pos}.

\begin{proof}[Proof of \cref{thm:multi-pos}]
	Combining Lemmas~\ref{lem:bounding-self} and \ref{lem:bounding-other} with Equations~\eqref{eq:welfare-decomposition} and~\eqref{eq:condition}, for every BNE $\sigma$, we have
	$$
	\left(2 + \max\{\gamma,c\}\cdot(\max\{\gamma,c+1\}+2)\right) \eq(\sigma) \ge \self+\other=\tildeopt\ge \opt/d.
	$$
	concluding the proof.
\end{proof}

We continue by proving Lemmas~\ref{lem:bounding-self} and~\ref{lem:bounding-other}, using two deviations which make use of smoothness-type inequalities (\cref{lem:multiple:bounded-utility}), as explained next.

\subsubsection{A Smoothness-type Lemma for Interdependence}
\label{sub:multi-pos-smooth-lemma}

Recall that in the standard smoothness framework for bounding PoA (with independent private values), a leading proof technique utilizes the fact that in equilibrium, agent $i$ cannot improve her utility by going ``all-in'' for one of the items~$\ell$.  
In interdependent settings, agents do not fully know their own values, so going all-in means the following (under no overbidding): participating only for item $\ell$ (i.e., reporting $\overline{a}_{i\ell}=1$ and $\overline{a}_{ik}=0$ for every item $k\neq\ell$),
and bidding $\overline{b}_{i\ell}=s_{i\ell}$ for this item (the bids for the other items can be arbitrary). 
To formulate a smoothness-type argument in \cref{lem:multiple:bounded-utility} below, we denote by $(\overline{b}_i,\overline{a}_i)$ the all-in report for agent $i$, and
let
$
j_\ell=j_\ell(\bids,\avec)
$ 
be the agent to whom the privatized second-price auction (\cref{alg:multiple-pos}) allocates item $\ell$ under the bidding profile $(\bids,\avec)$. If the item is not allocated at $(\bids,\avec)$ (no agents participate in the item's auction), then let $j_\ell(\bids,\avec)$ be a dummy agent with a constant value $0$ for item $\ell$. We now state our lemma:

\begin{lemma}
	\label{lem:multiple:bounded-utility}
	Consider agent $i$, bidding profile $(\bids,\avec)$, all-in bid $(\overline{b}_i,\overline{a}_i)$ for item $\ell$, and true signal profile $\s$. Then agent $i$'s utility $u_{i}(((\overline{b}_i,\bids_{-i}),(\overline{a}_i,\avec_{-i}));\s)$
	from going all-in for $\ell$ is at least
	\begin{enumerate} [(i)]
		\item $\hat{v}_{i\ell}(\s_{\ell}) - \max\limits_{j\neq i\ :\  a_{j\ell}=1}\{v_{j\ell}(b_{j\ell},\mathbf{0}_{-j\ell})\},$ \label{lb:self}
		
		\item $\tilde{v}_{i\ell}(\s_\ell)-\max\{\gamma,c+1\}\cdot v_{j_\ell\ell}(b_{j_\ell\ell},\s_{-j_\ell\ell})\cdot \I_{i\neq j_\ell}-v_{j_\ell\ell}(\s_\ell)\cdot \I_{i= j_\ell}.$  \label{lb:other}
	\end{enumerate}
\end{lemma}

\begin{proof}
	We begin by showing lower bound (\ref{lb:self}), and consider two cases:
	
	{\bf Case 1:} When going all-in agent $i$ is not allocated item $\ell$ (and thus is allocated no item). In this case we have that 
	$\max\limits_{j\neq i\ :\  a_{j\ell}=1} \{v_{j\ell}(b_{j\ell},\mathbf{0}_{-j\ell})\}\geq v_{i\ell}(s_{i\ell},\mathbf{0}_{-i\ell}) = \hat{v}_{i\ell}(\s_{\ell})$, 
	and since agent $i$'s utility is~$0$, lower bound (\ref{lb:self}) holds.
	
	{\bf Case 2:} When going all-in agent $i$ is allocated item $\ell$.  So her utility is
	$$v_{i\ell}(\s_\ell)- p_{i\ell}((\overline{b}_i,\bids_{-i}),(\overline{a}_i,\avec_{-i})) \geq  \hat{v}_{i\ell}(s_{\ell})-\max\limits_{j\neq i\ :\  a_{j\ell}=1}\{v_{j\ell}(b_{j\ell},\mathbf{0}_{-j\ell})\},$$ and (\ref{lb:self}) holds.
	
	We now proceed to lower bound (\ref{lb:other}). Observe that agent $i$'s utility from the all-in bid is nonnegative ($i$ pays at most her value for $\ell$ and does not compete for other items). Consider 
	agent $j_\ell=j_\ell(\bids,\avec)$ (the winner of item $\ell$ in the privatized second-price auction if agent $i$ reports $(b_i,a_i)$ rather than goes all-in). 
	If $j_\ell$ is in fact agent $i$ (so $\I_{i= j_\ell}=1$), then her utility from going all-in is
	\begin{eqnarray*}
		u_{i}(((\overline{b}_i,\bids_{-i}),(\overline{a}_i,\avec_{-i}));\s) \ \geq\ 0\ =\  \tilde{v}_{i\ell}(\s)-\tilde{v}_{j_\ell\ell}(\s)\ \geq\  \tilde{v}_{i\ell}(\s)-{v}_{j_\ell\ell}(\s).
	\end{eqnarray*} 
	
	It remains to show (\ref{lb:other}) if $j_\ell\ne i$. 
	We again consider the two cases from before:
	
	{\bf Case 1:} When going all-in agent $i$ is not allocated item $\ell$. We have that 
	\begin{flalign*}
		v_{j_\ell\ell}(b_{j_\ell\ell},\mathbf{0}_{-j_\ell\ell})\geq v_{i\ell}(s_{i\ell}, \mathbf{0}_{-i\ell}).
	\end{flalign*}
	Moreover, by $c$-single-crossing, we have that 
	$$c\left(v_{j_\ell\ell}(b_{j_\ell\ell},s_{i\ell},\mathbf{0}_{-\{ij_\ell\}\ell})-v_{j_\ell\ell}(s_{i\ell},\mathbf{0}_{-i\ell})\right)\ \geq\  v_{i\ell}(b_{j_\ell\ell},s_{i\ell},\mathbf{0}_{-\{ij_\ell\}\ell})-v_{i\ell}(s_{i\ell},\mathbf{0}_{-i\ell}).$$
	Combining the above two inequalities, we get that 
	\begin{eqnarray*}
		v_{i\ell}(b_{j_\ell\ell},s_{i\ell},\mathbf{0}_{-\{ij_\ell\}\ell}) &\leq & v_{j_\ell\ell}(b_{j_\ell\ell},\mathbf{0}_{-j_\ell\ell}) + c\left(v_{j_\ell\ell}(b_{j_\ell\ell},s_{i\ell},\mathbf{0}_{-\{ij_\ell\}\ell})-v_{j_\ell\ell}(s_{i\ell},\mathbf{0}_{-i\ell})\right)\\
		&\leq & (c+1)v_{j_\ell\ell}(b_{j_\ell\ell},s_{i\ell},\mathbf{0}_{-\{ij_\ell\}\ell}),
	\end{eqnarray*}
	where the last inequality is due to the monotonicity of valuation functions.
	By applying Corollary~\ref{col:j-bound-i}, we get that $$\max\{\gamma,c+1\}\cdot v_{j_\ell\ell}(b_{j_\ell\ell},\s_{-j_\ell\ell})\ \geq\  v_{i\ell}(b_{j_\ell\ell},\s_{-j_\ell\ell})\ \geq\  \tilde{v}_{i\ell}(\s).$$ 
	Since we are in the case where $i$ does not get the item, her utility is~$0$, and (\ref{lb:other}) holds. 
	
	{\bf Case 2:} When going all-in agent $i$ is allocated item $\ell$.  Notice that in this case, $\hat{v}_{j_\ell\ell}(\bb_\ell)=p_{i\ell}((\overline{b}_i,\bids_{-i}),(\overline{a}_i,\avec_{-i}))$ (by definition of $p_{i\ell}$ and since $j_\ell$ has the highest privatized value). We have that 
	\begin{eqnarray*}
		u_{i}(((\overline{b}_i,\bids_{-i}),(\overline{a}_i,\avec_{-i}));\s) & =  &v_{i\ell}(\s) - p_{i\ell}((\overline{b}_i,\bids_{-i}),(\overline{a}_i,\avec_{-i}))\\
		& =  &v_{i\ell}(\s) - \hat{v}_{j_\ell\ell}(\bb_\ell)
		\\ 
		&\geq& \tilde{v}_{i\ell}(\s)-v_{j_\ell\ell}(b_{j_\ell\ell},\s_{-j_\ell\ell}). 
	\end{eqnarray*}
	Thus, (\ref{lb:other}) holds.
\end{proof}

With these inequalities in hand, we now proceed to prove Lemmas~\ref{lem:bounding-self} and \ref{lem:bounding-other}.

\subsubsection{Bounding $\mathsf{SELF}$}
\label{sub:multi-pos-step2}

\begin{proof}[Proof of \cref{lem:bounding-self}]
	Consider the following deviation from agent $i$'s equilibrium strategy. Agent~$i$ samples an alternative signal profile $\bt\sim F,$ computes a matching $m^i=\widetilde{m}(s_i,\bt_{-i}),$ and goes all-in for the item $m^i_i$ she is matched to in $m^i$. As in \cref{sub:multi-pos-smooth-lemma}, the all-in report is $\overline{a}_{i\ell}=1$ if $\ell=m^i_i$ and $0$ otherwise, and $\overline{b}_{im^i_i}=s_{im^i_i}$.
	
	Our goal is to show $2\eq(\sigma)\geq\mathsf{SELF}$. We first notice that
	\begin{eqnarray}
	\eq &\geq& \sum_i\E_{\substack{\s\sim F\\(\bids,\avec)\sim \sigma(\s)}}\left[u_i((\bids,\avec);\s)\right]\nonumber\\
	&\geq & \sum_i\E_{\substack{\s,\bt\sim F\\(\bids,\avec)\sim\sigma(\s)\\m^i=\widetilde{m}(\s_i,\bt_{-i})}}\left[u_i\left((\overline{b}_i,\overline{a}_i,\bids_{-i},\avec_{-i}); \s\right)\right]\nonumber\\
	&\geq&\sum_i\E_{\substack{\s,\bt\sim F\\(\bids,\avec)\sim\sigma(\s)\\m^i=\widetilde{m}(\s_i,\bt_{-i})}}\left[v_{im_i^i}(s_{im_i^i},\mathbf{0}_{-im_i^i})-\max\limits_{j\neq i\ :\  a_{jm^i_i}=1}\{v_{jm_i^i}(b_{jm_i^i},\mathbf{0}_{-jm_i^i})\}\right]. \label{eq:self-eq-ineq}
	\end{eqnarray}
	In the above, the first inequality holds because the welfare is greater than the utility, the second holds because of the equilibrium hypothesis, and the third follows from lower bound~(\ref{lb:self}) in \cref{lem:multiple:bounded-utility}. 
	
	Consider the first term in Eq.~\eqref{eq:self-eq-ineq}:
	\begin{eqnarray}
	\sum_i\E_{\substack{\s,\bt\sim F\\m^i=\widetilde{m}(\s_{i},\bt_{-i})}} \left[v_{im_i^i}(s_{im_i^i},\mathbf{0}_{-im_i^i})\right] & =  &\sum_i\E_{\substack{\s\sim F\\m^i=\widetilde{m}(\s)}}\left[v_{im_i^i}(s_{im_i^i},\mathbf{0}_{-im_i^i})\right]\nonumber\\ &=& \E_{\s}\left[\sum_{(i,\ell)\in\widetilde{m}(\s)}v_{i\ell}(s_{i\ell},\mathbf{0}_{-i\ell})\right] \nonumber\\
	& = &\mathsf{SELF}\label{eq:self-welfare},
	\end{eqnarray}
	where the first equality is by renaming.
	
	As for the second term in Eq.~\eqref{eq:self-eq-ineq}, fixing an agent $i$,
	\begin{eqnarray}
	\E_{\substack{\s,\bt\sim F\\(\bids,\avec)\sim\sigma(\s)\\m^i=\widetilde{m}(\s_i,\bt_{-i})}}\left[\max\limits_{j\neq i\ :\  a_{jm_i^i}=1}\{v_{jm_i^i}(b_{jm_i^i},\mathbf{0}_{-jm_i^i})\}\right] & = & \E_{\substack{\s,\bt\sim F\\(\bids,\avec)\sim\sigma(\s_i,\bt_{-i})\\m^i=\widetilde{m}(\s)}}\left[\max\limits_{j\neq i\ :\  a_{jm_i^i}=1}\{v_{jm_i^i}(b_{jm_i^i},\mathbf{0}_{-jm_i^i})\}\right]\nonumber\\
	&= &\E_{\substack{\s,\bt\sim F\\(\bids,\avec)\sim\sigma(\bt)\\m^i=\widetilde{m}(\s)}}\left[\max\limits_{j\neq i\ :\  a_{jm_i^i}=1}\{v_{jm_i^i}(b_{jm_i^i},\mathbf{0}_{-jm_i^i})\}\right]\nonumber\\
	&\leq &\E_{\substack{\s,\bt\sim F\\(\bids,\avec)\sim\sigma(\bt)\\m^i=\widetilde{m}(\s)}}\left[\max\limits_{j\ :\  a_{jm_i^i}=1}\{v_{jm_i^i}(b_{jm_i^i},\mathbf{0}_{-jm_i^i})\}\right],\nonumber
	\end{eqnarray}
	where the first equality is by renaming and the second is because the term does not use $(\bids_i,\avec_i)$.
	Recall that $X_i(\bids,\avec)$ denotes the item subset agent $i$ receives under the bidding profile $(\bids,\avec)$. We have that 
	\begin{eqnarray}
	\sum_i\E_{\substack{\s,\bt\sim F\\(\bids,\avec)\sim\sigma(\s)\\m^i=\widetilde{m}(\s_i,\bt_{-i})}}\left[\max\limits_{j\neq i\ :\  a_{jm_i^i}=1}\{v_{jm_i^i}(b_{jm_i^i},\mathbf{0}_{-jm_i^i})\}\right]
	& \leq  &\E_{\substack{\s,\bt\sim F\\(\bids,\avec)\sim\sigma(\bt)\\m^i=\widetilde{m}(\s)}}\left[\sum_i\max\limits_{j\ :\  a_{jm_i^i}=1}\{v_{jm_i^i}(b_{jm_i^i},\mathbf{0}_{-jm_i^i})\}\right]\nonumber\\
	& = &\E_{\substack{\bt\sim F\\(\bids,\avec)\sim\sigma(\bt)}}\left[\sum_i\sum_{\ell\in X_i(\bids,\avec)} v_{i\ell}(b_{i\ell},\mathbf{0}_{-i\ell})\right]\nonumber\\
	& \leq & \E_{\substack{\bt\sim F\\(\bids,\avec)\sim\sigma(\bt)}}\left[\sum_i v_{i}(X_i(\bids,\avec);\bt)\right]\nonumber\\
	& = & \eq(\sigma)\label{eq:error-term-self},
	\end{eqnarray}
	where the first equality is derived from the allocation rule of simultaneous privatized second-price auctions, and the second inequality follows from the no-overbidding assumption.
	Combining Equations~\eqref{eq:self-eq-ineq}-\eqref{eq:error-term-self} yields the desired bound on $\mathsf{SELF}$.
\end{proof}

\subsubsection{Bounding $\mathsf{OTHER}$}
\label{sub:multi-pos-step3}

In this section, we prove Lemma~\ref{lem:bounding-other}. For this end, we use the following lemma:

\begin{lemma}
	\label{lem:arbitrary-covers-i}
	For every two agents $i$ and $j$, item $\ell$, and signal profile $\s$, $$\max\{\gamma,c\}\cdot\tilde{v}_{j\ell}(\s)\geq \tilde{v}_{i\ell}(\s)-v_{i\ell}(s_{i\ell},\mathbf{0}_{-i\ell}).$$
\end{lemma}

\begin{proof}
	Consider agent $j$ and let $k\neq j$ be some agent such that $\tilde{v}_{j\ell}(\s)=v_{j\ell}(\mathbf{0}_{k\ell},\s_{(-k)\ell})$. We consider two cases:
	
	{\bf Case~${k\neq i}$:} In this case, 
	\begin{flalign*}
		\tilde{v}_{j\ell}(\s) & \geq 	v_{j\ell}(\mathbf{0}_{k\ell},\s_{-k\ell})-v_{j\ell}(s_{i\ell},\mathbf{0}_{-i\ell})\\
		& \geq  \left(v_{i\ell}(\mathbf{0}_{k\ell},\s_{-k\ell})-v_{i\ell}(s_{i\ell},\mathbf{0}_{-i\ell})\right)/\max\{\gamma,c\} \\
		& \geq  \left(\tilde{v}_{i\ell}(\s)-v_{i\ell}(s_{i\ell},\mathbf{0}_{-i\ell})\right)/\max\{\gamma,c\},
	\end{flalign*}
	where the second inequality follows from Lemma~\ref{lem:j-bound-i-diff}, and the third inequality follows from the definition of truncated valuations.
	
	{\bf Case $k=i$:} In this case, 
	\begin{flalign*}
		\tilde{v}_{j\ell}(\s) & \geq 	v_{j\ell}(0_{i\ell},\s_{-i\ell})-v_{j\ell}(\mathbf{0})\\
		& \geq  \left(v_{i\ell}(0_{i\ell},\s_{-i\ell})-v_{i\ell}(\mathbf{0})\right)/\max\{\gamma,c\} \\
		& \geq  \left(v_{i\ell}(\s)-v_{i\ell}(s_{i\ell},\mathbf{0}_{-i\ell})\right)/\max\{\gamma,c\} \\
		& \geq  \left(\tilde{v}_{i\ell}(\s)-v_{i\ell}(s_{i\ell},\mathbf{0}_{-i\ell})\right)/\max\{\gamma,c\},
	\end{flalign*}
	where the second inequality follows from Lemma~\ref{lem:j-bound-i-diff}, and the third inequality follows from the submodularity in others' signals (Definition~\ref{cond:weak_submod}).\footnote{This claim is not true without this submodularity assumption, as demonstrated by the function $v_{j\ell}(s_i,s_j,s_k)=v_{i\ell}(s_i,s_j,s_k)=s_i\cdot s_j + s_i\cdot s_k$ with $s_i=s_j=s_k=1$. In this case, $\gamma=c=1$, and the LHS of the inequality in the lemma statement is 0, while the RHS is~1.}
\end{proof}

\begin{proof}[Proof of \cref{lem:bounding-other}]
	Consider the following deviation from agent $i$'s equilibrium strategy, where $i$ is among the first $m$ agents: Agent $i\in[m]$ goes all-in for item $i$. The all-in report is $\overline{a}_{i\ell}=1$ if $\ell=i$ and $0$ otherwise, and $\overline{b}_{ii}=s_{ii}$.	
	
	Our goal is to show $\max\{\gamma,c\}\cdot(\max\{\gamma,c+1\}+2)\cdot\eq(\sigma)\geq\mathsf{OTHER}$.  Recall that for item $\ell$, $j_\ell(\bids,\avec)$ is the agent to whom our mechanism allocates item $\ell$ under the bidding profile $(\bids,\avec)$.  We first notice that
	\begin{eqnarray}
	\eq &\geq& \sum_{i=1}^m\E_{\substack{\s\sim F\\(\bids,\avec)\sim \sigma(\s)}}\left[u_i((\bids,\avec);\s)\right]\nonumber\\
	&\geq & \sum_{i=1}^m\E_{\substack{\s\sim F\\(\bids,\avec)\sim \sigma(\s)}}\left[u_i\left((\overline{b}_i,\overline{a}_i,\bids_{-i},\avec_{-i}); \s\right)\right]\nonumber\\
	&\geq &\sum_{i=1}^m\E_{\substack{\s\sim F\\(\bids,\avec)\sim \sigma(\s)\\j_i=j_i(\bids,\avec)}}\left[\tilde{v}_{ii}(\s)-\max\{\gamma,c+1\}\cdot v_{j_ii}(b_{j_ii},\s_{-j_ii})\cdot \I_{i\neq j_i}-v_{j_ii}(\s)\cdot \I_{i= j_i}\right].\label{eq:eq-greater-other}
	\end{eqnarray}
	In the above, the first inequality holds because the welfare is greater than the utility, the second holds because of the equilibrium hypothesis, and the third follows from lower bound (\ref{lb:other}) in \cref{lem:multiple:bounded-utility}.  
	
	
	Consider the first term in Eq.~\eqref{eq:eq-greater-other}:
	\begin{eqnarray}
	\E_{\s}\left[\sum_{i=1}^m\tilde{v}_{ii}(\s)\right]&\geq& \E_{\s}\left[\sum_{(i,\ell)\in\widetilde{m}(\s)}\tilde{v}_{\ell\ell}(\s)\right]\nonumber\\&\geq& \E_{\s}\left[\sum_{(i,\ell)\in\widetilde{m}(\s)}\tilde{v}_{i\ell}(\s)-v_{i\ell}(s_{i\ell},\mathbf{0}_{-i\ell})\right]/\max\{\gamma,c\}\nonumber\\ &=& \mathsf{OTHER}/\max\{\gamma,c\},\label{eq:welfare-other}
	\end{eqnarray}
	where the second inequality follows from \cref{lem:arbitrary-covers-i}.
	
	As for the second term in Eq.~\eqref{eq:eq-greater-other}:
	\begin{eqnarray}
	& &\sum_{i=1}^m\E_{\substack{\s\sim F\\(\bids,\avec)\sim \sigma(\s)\\j_i=j_i(\bids,\avec)}}\left[\max\{\gamma,c+1\}\cdot v_{j_ii}(b_{j_ii},\s_{-j_ii})\cdot \I_{i\neq j_i}+v_{j_ii}(\s)\cdot \I_{i= j_i}\right]\nonumber
	\\&\leq &\max\{\gamma,c+1\}\E_{\substack{\s\sim F\\(\bids,\avec)\sim \sigma(\s)\\j_i=j_i(\bids,\avec)}}\left[\sum_{i=1}^mv_{j_ii}(b_{j_ii},\s_{-j_ii})\right] + \E_{\substack{\s\sim F\\(\bids,\avec)\sim \sigma(\s)\\j_i=j_i(\bids,\avec)}}\left[\sum_{i=1}^m v_{ii}(\s)\cdot \I_{i= j_i}\right]\nonumber\\
	& = & \max\{\gamma,c+1\}\E_{\substack{\s\sim F\\(\bids,\avec)\sim \sigma(\s)}}\left[\sum_i\sum_{\ell\in X_i(\bids,\avec)}v_{i\ell}(b_{i\ell},\s_{-i\ell})\right] + \E_{\substack{\s\sim F\\(\bids,\avec)\sim \sigma(\s)}}\left[\sum_{i=1}^m v_{ii}(\s)\cdot \I_{i\in X_i(\bids,\avec)}\right]\nonumber\\
	& \leq & \max\{\gamma,c+1\}\E_{\substack{\s\sim F\\(\bids,\avec)\sim \sigma(\s)}}\left[\sum_iv_i(X_i(\bids,\avec);\s)\right] + \E_{\substack{\s\sim F\\(\bids,\avec)\sim \sigma(\s)}}\left[\sum_{i=1}^m v_{i}(X_i(\bids,\avec);\s)\right]\nonumber\\
	& = & (\max\{\gamma,c+1\}+1)\eq(\sigma),\label{eq:error-term-other}
	\end{eqnarray}
	where the first inequality is due to the definition of $j_i(\bids,\avec)$ and the second inequality follows from the no-overbidding assumption.
	Combining Equations~\eqref{eq:eq-greater-other}-\eqref{eq:error-term-other} yields the desired bound on $\mathsf{OTHER}$.
\end{proof}



%

\section{Multiple Items: A Negative Result}
\label{sec:multiple-neg}

In this section, we show that for the case of multiple items, there is a separation between the case of few and many items; namely, when there are much more items then bidders, one cannot hope for a bound independent of $n$, even in the very simple case of (a) \textit{Common values} (which implies $\gamma=c=1$); (b) \textit{Symmetric valuations:} the valuation functions describing each of the items are identical, and each agent/item signal is drawn i.i.d.\footnote{Moreover, the valuations satisfy the diminishing marginal influence, Definition~\ref{cond:weak_submod}.}; (c) $\widetilde{\opt}/\opt\longrightarrow 1$ as $n$ grows large.

In our setting, each agent $i$ has a signal $s_{i\ell}$ for each item $\ell$, which is distributed according to $$s_{i\ell}=\begin{cases}
1\quad \text{w.p. } 1/n\\
0\quad \text{w.p. } 1-1/n
\end{cases},$$ independently (and identically) of other items. There are $n^2$ items, and the valuation function for each agent $i$ and item $\ell$ is $v_{i\ell}(\s)=1+\sum_j s_{j\ell}.$ 

In Theorem \ref{thm:multi-neg} we establish a gap between the optimal welfare and the welfare in equilibrium that is a function of $n$. Before presenting the result, we provide a high-level description of it along with some intuition.
\begin{itemize}
	\item Consider an arbitrary set of $n$ items. There are roughly $n$ high signals distributed on items in this set roughly uniformly. One can use a balls-and-bins type analysis to show that the highest-valued item in this set has value roughly $\ln n/\ln\ln n,$ up to constant factors (See Claim~\ref{clm:opt-lb}). Therefore, if one partitions the $n^2$ items arbitrarily into $n$ sets of size~$n$, and assigns the highest-valued item in each set to a different agent (this assignment can be arbitrary, due to common value), the obtained welfare is roughly $n\ln n/\ln\ln n.$
	\item Next, we bound the number of items an agent bids on by applying a no-overbidding assumption. Assume an agent bids on a set of items of size~$k$. Since the value of each item is at least~1, the sum of the agent's bids is at least~$k$. On the other hand, the value of any set of items is upper bounded by the value of the grand bundle, which is roughly~$\ln n/\ln\ln n$ (this is roughly the value of the highest-valued item in the grand bundle, according to a balls-and-bins type analysis) (see Claim~\ref{clm:expected-max-ub}). Therefore, the no-overbidding assumption implies that the expected number of items an agent bids on is $O(\ln n/\ln\ln n)$ (See Equation~\eqref{eq:bid-set-ub}).
	\item Finally, since an agent has a very vague idea about the valuable items (she only knows the $+1$ her signal may contribute to some of the items), the agent's bid is basically `a shot in the dark'. Using a balls-and-bins type analysis again, we show that the expected value an agent derives from the items she bids on is $O(\ln\ln n /\ln\ln\ln n)$ (see Claim~\ref{clm:expected-max-log2n-ub}), implying that $\eq = O(n\ln\ln n /\ln\ln\ln n)$, which yields a lower bound of $\tilde{\Omega}(\ln n).$\footnote{Here, $\tilde{\Omega}$ hides away terms that are $o(\ln n)$.}
\end{itemize}  


The main theorem in this section is the following. 

\begin{theorem} 
	\label{thm:multi-neg}
	There exists an instance with $n$ unit-demand agents and $n^2$ items, with $\gamma=c=1$, and $\widetilde{\opt}\approx\opt$, where under a no-overbidding assumption, for every item-bidding mechanism and every equilibrium $\sigma$,
	$\opt/\eq(\sigma)=\Omega(\ln n\ln\ln\ln n/\ln\ln^2 n)$.
\end{theorem}
\vspace{-2pt}

Before proving the theorem, we provide some useful claims. 
We first establish a lower bound on $\opt$.

\begin{claim}
	 \label{clm:opt-lb}
	$\opt=\Omega(n\ln n/\ln\ln n)$.
\end{claim}
	
\begin{proof}
	We partition the items into $n$ sets of $n$ items. It suffices to show that in each of the sets, $\E_\s[\max_{\ell\in [n]}v_{i\ell}(\s)]=\Omega(\log n/\log\log n)$, since we can allocate the best items in each set to a different bidder. 
	
	The probability that an item's value is at least $k+1$ is 
	\begin{eqnarray}
	\Pr_{\s}\left[ v_{i\ell}(\s)\geq k+1\right] &=& \Pr_\s\left[\sum_{i\in [n]}s_{i\ell}\geq k\right]\nonumber\\
	&\geq&  \Pr_\s\left[\sum_{i\in [n]}s_{i\ell}= k\right]\nonumber\\
	& = & \binom{n}{k}\left(\frac{1}{n}\right)^k\left(1-\frac{1}{n}\right)^k\nonumber\\
	&\geq& \left(\frac{n}{k}\right)^k\frac{1}{n^k}\frac{1}{e}\nonumber\\
	& = &\frac{1}{ek^k}.
	\end{eqnarray}
	
	Setting $k=\frac{\ln n}{3\ln\ln n}$, we have $k^k\leq \ln n^{\frac{\ln n}{3\ln\ln n}}=e^{\frac{\ln n}{3}}=n^{1/3}.$ Therefore, $$\Pr_{\s}\left[ v_{i\ell}(\s)\geq \frac{\ln n}{3\ln\ln n}+1\right]\geq \frac{1}{en^{1/3}}.$$
	Let $X_{\ell}$ be the indicator for the event that item $\ell$ has value greater than $\frac{\ln n}{3\ln\ln n}$, and $X$ the number of such items in a set of $n$ items. We know that $$\mu=\E_{\s}\left[X\right] \geq n\cdot \frac{1}{en^{1/3}}=n^{2/3}/e.$$
	By Chebychev inequality, we have that 
	\begin{eqnarray*}
	\Pr_\s\left[\max_{\ell\in [n]}X_\ell \leq \frac{\ln n}{3\ln\ln n}+1\right] & \leq & \Pr_{\s}\left[\lvert X-\mu\rvert\geq \mu\right] \\
	&\leq&  \frac{\Var(X)}{\mu^2}\\
	& \leq&  e^2\cdot \frac{\sum_{\ell=1}^n(\Var(X_\ell))+\sum_{\ell\neq j}\Cov(X_\ell,X_j)}{n^{4/3}}. 
	\end{eqnarray*}
	For each $\ell$, $\Var(X_\ell)\leq \E[X_\ell]^2= \frac{1}{e^2n^{2/3}},$ and since $X_\ell$ and $X_j$ are independent for each $\ell\neq j$, we also have that $\Cov(X_\ell,X_j)=0$ for each $\ell\neq j$. Therefore, we get $\Pr_\s\left[\max_{\ell\in [n]}X_\ell \leq \frac{\ln n}{3\ln\ln n}+1\right] \leq 1/n,$ 
	which implies that $\E_\s[\max_{\ell\in [n]}v_{i\ell}(\s)]=\Omega(\ln n/\ln\ln n)$. This concludes the proof of the claim.	
\end{proof}

The following claim is used to bound the expected value of the maximal item, and thus the number of items an agent bids on.
\begin{claim}
	\label{clm:expected-max-ub}
	For any $i$, and for any realization of $s_i=(s_{i1},\ldots,s_{im})$,
$$\E_{\s_{-i}}\left[\max_\ell v_{i\ell}(\s)\ \Big\vert\ s_i\right]\leq 4\ln n/\ln\ln n+3.$$
\end{claim}

\begin{proof}
	Consider a bidder $i$. We have that 
	\begin{eqnarray}
	\Pr_{\s_{-i}}\left[\max_\ell v_{i\ell}(\s)\geq k+2\ \Big\vert\ s_i\right] & = & \Pr_{\s_{-i}}\left[\max_\ell\{ 1+\sum_{j\in[n]}s_{j\ell}\}\geq k+2\ \Big\vert\ s_i\right]\nonumber\\
	&\leq  &\Pr_{\s_{-i}}\left[\max_\ell\{\sum_{j\neq i}s_{j\ell}\}\geq k\right]\nonumber\\
	& =  &\Pr_{\s_{-i}}\left[\exists \ell\ :\ \sum_{j\neq i}s_{j\ell}\geq k\right]\nonumber\\
	&\leq&  \Pr_{\s}\left[\exists \ell\ :\ \sum_{j\in [n]}s_{j\ell}\geq k\right]\nonumber\\
	&\leq & n^2\cdot \Pr_{\s}\left[ \sum_{j\in [n]}s_{j\ell}\geq k\right].\label{eq:highmax-bound}
	\end{eqnarray}
	
	The first inequality follows because the signals are sampled independently, and the best case is where $s_i$ is the all ones vector, and the last inequality is by applying a union bound. We can use a balls and bins type analysis in order to bound $\Pr_{\s}\left[ \sum_{j\in [n]}s_{j\ell}\geq k\right]$. Namely, If we consider any $k$ agents, the probability that these $k$ agents have a high signal for item $\ell$ is $\left(\frac{1}{n}\right)^k$. Using a union bound, we get that the probability there exist $k$ agents with a high signal for item $\ell$ is at most $\binom{n}{k}\left(\frac{1}{n}\right)^k$. Therefore, 
	\begin{eqnarray*}
		\Pr_{\s}\left[ \sum_{j\in [n]}s_{j\ell}\geq k\right] &\leq &\binom{n}{k}\left(\frac{1}{n}\right)^k\\
		&\leq &\left(\frac{en}{k}\right)^k\left(\frac{1}{n}\right)^k\\
		& =  &\left(\frac{e}{k}\right)^k,
	\end{eqnarray*}
	where the second inequality follows Stirling's approximation. Setting $k=4\ln n/\ln\ln n$, we have that 
	\begin{eqnarray*}
		\Pr_{\s}\left[ \sum_{j\in [n]}s_{j\ell}\geq k\right] &\leq &\left(\frac{e\ln\ln n}{4\ln n}\right)^\frac{4\ln n}{\ln\ln n}\\
		& \leq &\exp\left(\frac{4\ln n}{\ln\ln n}(\ln\ln\ln n-\ln\ln n)\right)\\
		& =  &\exp\left(-4\ln n + \frac{4\ln n\ln\ln\ln n}{\ln\ln n}\right),
	\end{eqnarray*}
	and for $n$ large enough,
	$$\Pr_{\s}\left[ \sum_{j\in [n]}s_{j\ell}\ \geq\ k\right]\ \leq\ \exp\left(-3\ln n\right) = \frac{1}{n^3}.$$
	Therefore, by Equation~\eqref{eq:highmax-bound},
	\begin{eqnarray}
	\Pr_{\s_{-i}}\left[\max_\ell v_{i\ell}(\s)\ \geq\ \frac{4\ln n}{\ln\ln n}+2\ \Big\vert\ s_i\right]\ \leq\  1/n.
	\end{eqnarray}
	We get that 
	\begin{eqnarray*}
		\E_{\s_{-i}}\left[\max_\ell v_{i\ell}(\s)\ \Big\vert\ s_i\right]& \leq &\Pr_{\s_{-i}}\left[\max_\ell v_{i\ell}(\s)< \frac{4\ln n}{\ln\ln n}+2\ \Big\vert\ s_i\right]\cdot \left(\frac{4\ln n}{\ln\ln n}+2\right) \\
		&  &\ +\ \Pr_{\s_{-i}}\left[\max_\ell v_{i\ell}(\s)\geq \frac{4\ln n}{\ln\ln n}+2\ \Big\vert\ s_i\right]\cdot (n+1) \\ 
		&\leq  &(1-1/n)\cdot \left(\frac{4\ln n}{\ln\ln n}+2\right) + (n+1)/n\\
		& \leq  &\frac{4\ln n}{\ln\ln n}+3.
	\end{eqnarray*}
	
\end{proof}

We also bound the expected value an agent might derive by bidding on a set of $\ln ^2n$ items. 
\begin{claim}
	\label{clm:expected-max-log2n-ub}
	If agent $i$ bids on a set of at most $\ln ^2n$ items, her expected value for that set is at most $\frac{4\ln\ln n }{\ln\ln\ln n}+4$ for any realization of $s_i$.
\end{claim}

\begin{proof}
	We first bound the probability that the agent sees an item with value greater than $\frac{4\ln\ln n }{\ln\ln\ln n}+2$. Notice that an agent chooses a set of items to bid on without knowing anything about the realization of the signals of other bidders. That is, if $S_i$ is that set of items $i$ bids on,
	\begin{eqnarray*}
		\Pr_{\s_{-i}}\left[\max_{\ell\in S_i} v_{i\ell}(\s) \geq k+2 \Big| s_i\right] &\leq & \Pr_{\s_{-i}}\left[\max_{\ell\in S_i} \sum_{j\neq i}s_{j\ell} \geq k\right]\\
		&\leq & |S_i|\Pr_{\s}\left[\sum_j s_{j\ell}\geq k\right]\\
		& \leq & \ln^2n\cdot\Pr_{\s}\left[\sum_j s_{j\ell}\geq k\right],
	\end{eqnarray*}
	where the derivations are similar to the ones in Claim~\ref{clm:expected-max-ub}.
	
	Using $k=\frac{4\ln\ln n }{\ln\ln\ln n},$ a similar progression to the one in Claim~\ref{clm:expected-max-ub} shows that 
	$$\Pr_{\s}\left[\sum_j s_{j\ell}\ \geq\ \frac{4\ln\ln n }{\ln\ln\ln n}\right]\  \leq\ \ln^{-3} n,$$
	which yields $$\Pr_{\s_{-i}}\left[\max_{\ell\in S_i} v_{i\ell}(\s)\ \geq\ \frac{4\ln\ln n }{\ln\ln\ln n}+2 \Big| s_i\right]\ \leq\ \frac{1}{\ln n}.$$
	We conclude by noting that 
	\begin{eqnarray*}
		\E_{\s_{-i}}\left[\max_{\ell\in S_i} v_{i\ell}(\s) \right] &\leq& \Pr_{\s_{-i}} \left[\max_{\ell\in S_i} v_{i\ell}(\s) \leq \frac{4\ln\ln n }{\ln\ln\ln n}+2 \Big| s_i\right]\cdot  \left(\frac{4\ln\ln n }{\ln\ln\ln n}+2\right) \\
		& &\ + \Pr_{\s_{-i}} \left[\frac{4\ln\ln n }{\ln\ln\ln n}+2 \leq \max_{\ell\in S_i} v_{i\ell}(\s) \leq \frac{4\ln n }{\ln\ln n}+2 \Big| s_i\right]\cdot  \left(\frac{4\ln n }{\ln\ln n}+2\right)\\
		& &\ + \Pr_{\s_{-i}}\left[\max_{\ell\in S_i} v_{i\ell}(\s)\geq \frac{4\ln n }{\ln\ln n}+2 \Big| s_i\right]\cdot  \left(n+1\right)\\
		&\leq & \left(\frac{4\ln\ln n }{\ln\ln\ln n}+2\right) + \frac{1}{\ln n}\left(\frac{4\ln n }{\ln\ln n}+2\right) + \frac{1}{n}(n+1)\\
		&\leq & \frac{4\ln\ln n }{\ln\ln\ln n} +2 +4/\ln\ln n+2/\ln n + 1 + 1/n\\
		&\leq &\frac{4\ln\ln n }{\ln\ln\ln n}+4,
	\end{eqnarray*} 
	which completes the proof.
\end{proof}

With these claims in hand, we are ready to establish the proof of Theorem~\ref{thm:multi-neg}.


\begin{proof}[Proof of Theorem~\ref{thm:multi-neg}]
	Let $S_i$ be the random set of items agent $i$ bids on. We can bound the expected size of $S_i$ for every $s_i$ by
	\begin{flalign}
	\E[\lvert S_i\rvert \mid s_i]\ \leq\ \E_{\substack{\s_{-i}\sim\F_{-i}\\(b_i,a_i)\sim \sigma_i(s_i)}}\left[\sum_{\ell\ :\ a_{i\ell}=1}v_{i\ell}(b_{i\ell},\s_{-i})\right]\ \leq\ \E_{\s_{-i}\sim\F_{-i}}\left[v_{i}(S_i;\s)\right]\ \leq\ \frac{4\ln n}{\ln\ln n}+3.\label{eq:bid-set-ub}
	\end{flalign}
	The first inequality in \eqref{eq:bid-set-ub} follows by the fact that the value is at least 1 for every item and agent, the second inequality follows by no-overbidding, and the third inequality follows by Claim~\ref{clm:expected-max-ub}.
	
	By Markov inequality, and by Equation~\eqref{eq:bid-set-ub}, we have that for every $s_i$, 
	\begin{flalign*}
		\Pr\left[\lvert S_i\rvert> \ln^2n \mid s_i\right] < \frac{\frac{4\ln n}{\ln\ln n}+3}{\ln^2n}=O(1/\ln n\ln\ln n).
	\end{flalign*}
	We now get that for every $s_i$

	\begin{eqnarray*}
		\E_{\s_{-i}, S_i}\left[v_i(S_i;\s) \mid s_i\right]&\leq &\Pr\left[|S_i| \leq \ln^2 n \mid s_i\right]\left(\frac{4\ln\ln n }{\ln\ln\ln n}+4\right)\ +\ \Pr\left[|S_i| > \ln^2 n \mid s_i\right]\cdot \left(\frac{4\ln n}{\ln\ln n}+3\right)\\
		& = & O\left(\frac{\ln\ln n}{\ln\ln\ln n}\right)+ O\left(\frac{1}{\ln n\ln\ln n}\frac{\ln n}{\ln\ln n}\right) \\ & = &  O\left(\frac{\ln\ln n}{\ln\ln\ln n}\right),
	\end{eqnarray*}
	where the first inequality follows by Claims~\ref{clm:expected-max-ub} and~\ref{clm:expected-max-log2n-ub}.
	
	Therefore, the expected value of the set of items an agent $i$ bids on is $\E_{\s, S_i}\left[v_i(S_i;\s)\right]= O\left(\frac{\ln\ln n}{\ln\ln\ln n}\right)$. Since this is an upper bound on the value of items the agent wins, every agent gets a value of $O\left(\frac{\ln\ln n}{\ln\ln\ln n}\right)$ in an equilibrium, thus $\eq(\sigma)= O(n\ln\ln n/\ln\ln\ln n)$. Since $\opt= \Omega\left(n\ln n/\ln\ln n\right)$ (by Claim~\ref{clm:opt-lb}), we conclude that $\opt/\eq = \Omega(\ln n\ln\ln\ln n/\ln\ln^2 n)$.
	
	This concludes the proof Theorem~\ref{thm:multi-neg}.
\end{proof}

\paragraph{Acknowledgements.} We deeply thank David Parkes for helpful discussions.
\bibliographystyle{plainnat}
\bibliography{abb,interdep-bib}
\newpage

\appendix

\section{Appendix for Section \ref{sec:intro}}
\label{appx:intro}
\subsection{No Pure Equilibrium Example}

\begin{proposition}
	\label{pro:no-pure-equil}
	There exists a single-item, $2$-agent setting with $c$-SC interdependent valuations such that every auction that allocates the item to the highest-valued bidder, and charges at most her reported value, admits no pure Nash equilibrium under the no overbidding assumption.
\end{proposition}

\begin{corollary}
	There exists a single-item 2-agent setting with $c$-SC interdependent valuations such that neither the second-price auction nor the generalized Vickrey auction admit a pure Nash equilibrium under the no overbidding assumption.
\end{corollary}


\begin{proof}[Proof of \cref{pro:no-pure-equil}]
	Let ${\mathcal M} = (x,p)$ be a single-item auction that allocates the item to the bidder with the highest reported value, and charges at most the reported value.
	Consider a 2-bidder setting, with signal spaces ${\mathcal S}_1 = {\mathcal S}_2 = [0,3\pi]$ and the following valuations (depicted in Figure \ref{fig:noPNEvaluation} as a function of $s_1+s_2$):
	$v_1(\bs) = \sin(s_1 + s_2) + 2(s_1+s_2)$ and 
	$v_2(\bs) = \sin(s_1 + s_2 +\pi) + 2(s_1+s_2)$.
	
	One can easily verify that these valuations are monotone in both signals, and c-SC for $c=3$.
	We claim that for the case where the signal profile is $\bs = (2\pi, 2\pi)$, there exists no pure Nash equilibrium satisfying no overbidding.
	Falsely assume that there exists $\bids \leq \bs$ such that $\bids$ is a pure Nash equilibrium.
	We show that for every bid $b_i$ of the winning agent, the losing agent has a beneficial deviation.
	Suppose $\bids$ is such that agent $1$ wins the item. By no overbidding, $b_1 \leq s_1 =2\pi$. We claim that agent 2 has a deviation that grants her strictly positive utility. We first show that there exists $b'_2 \leq s_2 =2\pi$ such that agent 2 wins the item. Indeed, there exists $k\in\{1,2\}$ such that for $b'_2$ satisfying $b_1 + b'_2 = 2\pi k - \pi/2$, it holds that $v_2(b_1,b'_2) = 2(2\pi k - \pi/2)+1 > 2(2\pi k - \pi/2) - 1 = v_1(b_1,b'_2)$. By the allocation rule, agent $2$ wins the item.
	Furthermore, by the assumption on the payment, $p_2(b_1,b'_2) \leq v_2(b_1,b'_2)$. Therefore, 
	$u_2((b_1,b'_2);\s) = v_2(\bs) - p_2(b_1,b'_2) \geq v_2(\bs) - v_2(b_1,b'_2) = 8\pi - (2(2\pi k - \pi/2)+1) \geq \pi - 1 >0$. 
	Thus, agent 2 obtains positive utility.
	
	Analogously, if $\bids$ is such that bidder $2$ wins the item, bidder 1 has a deviation $b'_1 \leq s_1$ that grants her positive utility (i.e., there exists $k\in\{0,1\}$ such that for $b'_1$ satisfying $b'_1 + b_2 = 2\pi k + \pi/2$, bidder 1 has strictly positive utility). 
	This concludes the proof.	
	\begin{figure}[H]	
		\centering
		\includegraphics[width=0.6\linewidth]{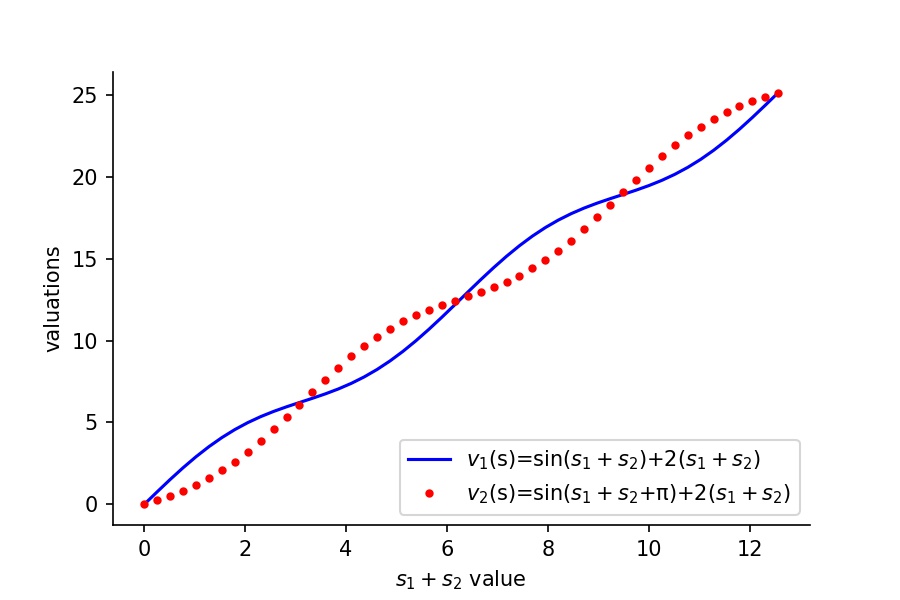}
		\caption{Valuations of agents $1$ and $2$ as functions of $s_1+s_2$. For every bid of one agent, the other  has a winning bid.}
		\label{fig:noPNEvaluation}
	\end{figure}
\end{proof}

\subsection{Appendix for ``Challenges'' Subsection}
\label{appx:challenges}

\begin{proposition*}[\bf Formal statement of Proposition~\ref{clm:GVA:POA}]
	There exists a single-item, $n$-bidder setting, satisfying single-crossing, such that the $\epoa$ of every auction that allocates the item to the highest-valued bidder is $\Omega(n)$ (even under no-overbidding).
\end{proposition*}

\begin{proof}[Proof of \cref{clm:GVA:POA}]
	Consider a single-item setting with $n$ bidders, signal spaces $S_1 = \{1\}$, $S_2 = \{1\}$ and $S_i = [0, 1]$ for every $i \geq 3$, and the following valuation profile:
	\begin{equation}
	v_1 = \sum_{i\in [n]} s_i +\epsilon,~~~
	v_2 = 2(s_2+\epsilon),~~~\text{and}~~~
	\forall i \geq 3: v_i = s_i,\label{eq:bad-example1}
	\end{equation}
	where $\epsilon>0$ is arbitrarily small.
	Observe that the valuation profile satisfies single-crossing and submodularity over signals~\citep{EdenFFGK19}.
	We argue that the bidding strategy profile $\sigma$ where $\sigma_1(s_1) =1$, $\sigma_2(s_2)=1$ and $\sigma_i(s_i)=0$ for every $i \geq3$ is an ex-post equilibrium (EPE):
	Given the reported signals we have $v_1(\sigma(\s))=2+\epsilon$, $v_2(\sigma(\s))=2+2\epsilon$, and $v_i(\sigma(\s))=0$ for every $i\geq 3$. Thus bidder $2$ wins the item. Bidders $1,2$ cannot deviate since their signal spaces are singletons, and no bidder $i\geq 3$ can ever win by deviating since $v_1(\bids')>v_i(\bids')$ for every bid profile $\bids'$. 
	If the true signal profile is $\bs=(1,\dots,1)$ then the optimal welfare is $n + \epsilon$, while the welfare achieved at EPE $\sigma$ is $2+2\epsilon$, completing the proof.  
\end{proof}

\begin{proposition*}[\bf Formal statement of Proposition~\ref{pro:POA-EPIC}]
	There exists a single-item, $n$-bidder setting, satisfying single-crossing, such that the $\epoa$ of every deterministic ex-post IC-IR mechanism is $\Omega(\sqrt{n})$ (even under no-overbidding).
\end{proposition*}

\begin{proof}
	Consider a single-item setting with $n$ bidders, signal spaces $S_i=[0,1]$ for every $i$, and the following valuation profile:
	\begin{equation}
	v_1 = \sum_{i\in [n]} s_i +\epsilon,~~~
	v_2 = \sqrt{n}s_2,~~~\text{and}~~~
	\forall i \geq 3: v_i = s_i,\label{eq:bad-example2}
	\end{equation}
	where $\epsilon>0$ is arbitrarily small.
	Observe that the valuation profile satisfies single-crossing and submodularity over signals~\citep{EdenFFGK19}.
	We argue that the bidding strategy profile $\sigma$ where
	$\sigma_1(s_1) = s_1$, $\sigma_2(s_2) = s_2$ and $\sigma_i(s_i) = 0$ for every $i\geq 3$ is an EPE under no-overbidding: 
	No bidder $i\geq 3$ can ever win by deviating since $v_1(\bids')>v_i(\bids')$ for every bid profile $\bids'$.
	Bidders $i=1,2$ who are truthful cannot increase their bids (no overbidding).
	
	A winning bidder cannot gain by decreasing her bid. She either loses the item, in which case, her utility is 0, while her utility before the deviation is non-negative by no-overbidding and ex-post IR property; or still wins the item, in which case her payment is the same by the payment rule of ex-post IC-IR mechanisms, and so is her utility. 
	
	A losing bidder cannot win the item by decreasing her bid since the valuation profile satisfies single-crossing and due to the characterization of ex-post IC-IR mechanisms as monotone. Thus $\sigma$ is an EPE as claimed.
	
	Suppose the true signal profile is ${\bf s^1} = (1,1,0,\dots,0)$. The only nonzero values are then $v_1(\s^1)=2+\epsilon,v_2(\s^1)=\sqrt{n}$. If all bidders report truthfully, this is an ex-post equilibrium since the mechanism is ex-post IC. If the item is allocated to any bidder but 2, then the proposition statement follows. Assume therefore that given bid profile $\s^1$ the mechanism allocates to bidder~$2$. 
	
	Suppose now the true signal profile is $\s^2 = (1,\dots, 1)$. Since $\sigma$ is an EPE, if agents bid according to $\sigma$, they bid $\sigma(\s^2)=\s^1$, and so bidder 2 is the winner. The obtained welfare is $v_2(\s^2)=\sqrt{n}$, whereas the optimal welfare is $v_1(\s^2)=n+\epsilon$. This concludes the proof.
\end{proof}

\section{Appendix for Section \ref{sec:prelim}}

\subsection{Appendix for ``PoA Background'' Subsection}

\begin{proof}[Proof of Observation~\ref{obs:equil-hierarchy}]
	First assume $\sigma$ is an $\epe$. By definition, for every $i$ and $\bs'_{-i}$, we have that $\sigma_i(s_i)$ maximizes $u_i(\sigma(s_i,\s'_{-i});(s_i,\s'_{-i}))$. In particular this holds for $\s_{-i}$. Therefore, $\sigma$ is a $\pne$ with respect to $\s$. Now assume that $\sigma$ is a $\pne$ (not necessarily an $\epe$). 
	Consider a point-mass distribution $F$ with probability 1 for signal profile $\bs$. 
	Observe that $\sigma$ is a $\bne$ for this Bayesian setting.
\end{proof}

\label{appx:poa}

\section{Appendix for Section~\ref{sec:single} (Single Item)}
\label{appx:single}

\subsection{Improved PoA Bound for Generalized Vickrey Auction}

\begin{proposition}[Improved bound]
	\label{prop:GVA:gamma-bound-sc:gamma-BPOA}
	Consider a single-item setting with $\gamma$-heterogeneous, single crossing, continuous valuations.
	The $\bpoa$ of the generalized Vickrey auction under no-overbidding is at most~$\gamma$.
\end{proposition}	

In order to prove \cref{prop:GVA:gamma-bound-sc:gamma-BPOA}, we first establish the following claim.

\begin{claim} \label{clm:gamma-bound:GVA:single:BNE} 
	Let $\sigma$ be a BNE of a single-item generalized Vickrey auction with single crossing, $\gamma$-heterogeneous, continuous valuations which satisfies no-overbidding condition. For every $\bs \sim \F$ and $\bids \sim \sigma(\bs)$, let $i_{\s}$ and $i_{\bids}$ be some players with the highest values under $\bs$ and $\bids$ respectively.
	Then $\gamma v_{i_{\bids}}(\bs) \geq v_{i_{\s}}(\bs)$.
\end{claim}

\begin{proof}
	For some $\s\sim \F$, let $\bids$ be some bid profile in the support of $\sigma(\s)$.
	If $i_\s = i_\bids$ the claim is trivial, thus it is sufficient to consider the case where $i_{\s}$ is not allocated in $\bids$.
	Assume $i_{\s}$ is not allocated in $\bids$, hence, $u_{i_\bs}(\bids) = 0$.
	
	First we show that $u_{i_{\s}}(s_{i_{\s}}, \bids_{-{i_{\s}}}) = 0$. Assume toward contradiction that this is not the case, we define a new strategy profile for $i_\bs$ $$\sigma'_{i_{\s}}(t_i)=\begin{cases}\sigma_{i_{\s}}(t_{i_{\s}})&\quad t_{i_{\s}}\neq s_{i_{\s}} \\ s_{i_{\s}} &\quad t_{i_{\s}}= s_{i_{\s}}\end{cases}.$$
	
	Note that for all signal profile $\s'$  agent's $i_s$ bid under the new strategy profile can only increase with respect to her bid under the old strategy profile, thus if agent $i_{\s}$ wins under $\sigma(\s')$, then by single-crossing, she remains the winner. 
	Moreover, by the critical payment rule, her payment remains the same, thus her utility remains the same.
	If agent $i_{\s}$ does not win under $\sigma(\s')$, then her utility is $0$. When altering her strategy to $\sigma'_{i_\s}$, then by no-overbidding, the allocation rule and the critical payment rule, her utility is non-negative, thus it can not decrease. Moreover, when $\s'=\s$, her utility strictly increases.
	Thus, $\Exp_{\s\sim\F}[u_{i_{\s}}(\sigma'_{i_{\s}}(\s), \sigma_{-{i_{\s}}}(\s))] > \Exp_{\s\sim\F}[u_{i_{\s}}(\sigma(\s))],$ in contradiction to $\sigma$ being a BNE.
	
	We proceed with our main proof.
	Consider two cases:
	\begin{enumerate}
		\item $i_{\s}$ does not win the item under $(s_{i_{\s}}, \bids_{-{i_{\s}}})$. There exists an agent $j\neq i_\s$ s.t. $v_j(s_{i_{\s}}, \bids_{-{i_{\s}}}) \geq  v_{i_{\s}}(s_{i_{\s}}, \bids_{-{i_{\s}}})$; $i_{\bids}$ has the highest value under $\bids$, so $v_{i_{\bids}}(\bids) \geq v_j(\bids)$ and from Corollary~\ref{col:j-bound-i} we get $\gamma v_{i_{\bids}}(s_{i_{\s}}, \bids_{-{i_{\s}}}) \geq v_j(s_{i_{\s}}, \bids_{-{i_{\s}}}) \geq v_{i_{\s}}(s_{i_{\s}}, \bids_{-{i_{\s}}})$, ($i_{\bids},j \neq i_{\s}$ and $b_{i_{\s}} \leq s_{i_{\s}}$) by applying Corollary~\ref{col:j-bound-i} again, we get $\gamma v_{i_{\bids}}(\s) \geq v_{i_{\s}}(\s)$.

		\item $i_{\s}$ wins the item under $(s_{i_{\s}}, \bids_{-{i_{\s}}})$. Since $i_{\s}$'s utility is~0, then $v_{i_{\s}}(\s)= p_{i_{\s}}(s_{i_{\s}}, \bids_{-{i_{\s}}}) = v_{i_{\s}}(b^{*}_{i_{\s}}, \bids_{-{i_{\s}}}),$ where $b^{*}_{i_{\s}}$ is her critical bid. Since the valuation functions are continuous, and $i_{\s}$ does not win under $\bids$, there exists an agent $j\neq i$ such that $v_j(b^{*}_{i_{\s}}, \bids_{-{i_{\s}}}) = v_{i_{\s}}(b^{*}_{i_{\s}}, \bids_{-{i_{\s}}})$. By the same arguments as above, $\gamma v_{i_{\bids}}(b^{*}_{i_{\s}}, \bids_{-{i_{\s}}}) \geq v_j(b^{*}_{i_{\s}}, \bids_{-{i_{\s}}}) = v_{i_{\s}}(\s).$ By the monotonicity of the valuation functions, we get $\gamma v_{i_{\bids}}(\s) \geq v_{i_{\s}}(\s)$.
	\end{enumerate}
	This concludes the proof.
\end{proof}

We are now ready to prove \cref{prop:GVA:gamma-bound-sc:gamma-BPOA}.

\begin{proof}[Proof of \cref{prop:GVA:gamma-bound-sc:gamma-BPOA}] 
	Let $\sigma$ be a BNE of a generalized Vickrey auction with the above properties. Denote $\F$ the distribution over the signals of the agents; for every signal profile ${\s} \sim \F$ and bid profile $\bids \sim \sigma({\s})$ denote $i_{\s}$ and $i_{\bids}$ the agents with the highest values under $\s$ and $\bids$ respectively.
	
	By \cref{clm:gamma-bound:GVA:single:BNE}, $\gamma v_{i_{\bids}}(\s) \geq v_{i_{\s}}(\s)$.
	It holds that:
	\begin{flalign}
	\opt = \Exp_{\s\sim\F}[OPT(\s)] & = \Exp_{\s\sim\F}[v_{i_{\s}}(\s)] = \Exp_{\s\sim\F}[\Exp_{\bids\sim\sigma(\s)}[v_{i_{\s}}(\s)]]\nonumber\\
	&\leq \Exp_{\s\sim\F}[\Exp_{\bids\sim\sigma(\s)}[\gamma v_{i_{\bids}}(\s)]] = \gamma \Exp_{\bs\sim\F}[\Exp_{\bids\sim\sigma(\s)}[v_{i_{\bids}}(\s)]] \label{eq:gamma-bound:single}\\
	&= \gamma \Exp_{\s\sim\F}[\SW(\sigma,\s)] = \gamma\eq,\nonumber
	\end{flalign}
\end{proof}

\subsection{PoA Bound for Second-Price Auction}
\label{appx:2PA}

We now proceed to prove the same general upper bound from Theorem~\ref{prop:upper-bound:gamma-bound+c-sc:GVA} for second-price for interdependence:

\begin{theorem}\label{thm:upper-bound:gamma-bound+c-sc:2PA}
	Consider a single-item setting with $\gamma$-heterogeneous, $c$-SC valuations.	
	The $\bpoa$ of the generalized second-price auction for interdependent values under no-overbidding is bounded by $1+\max\{\gamma,c\}$. 
\end{theorem}

\begin{proof}
	Note, the allocation rule of the second-price auction and the generalized Vickrey auction is the same, and they differ only in the payment rule.
	Moreover, the only part in the proof of Theorem~\ref{prop:upper-bound:gamma-bound+c-sc:GVA} where we used the payment rule is in case 2, thus, it is suffice to prove case 2 with the second-price payment rule, and the proof will hold for this Theorem as well.
	
	\noindent {\bf Case 2:} Bidder $i$ wins the item under the bidding profile $(s_i,\bids_{-i})$. Let $j \in \argmax_{l\neq i} \{v_l(s_i,\bids_{-i})\}$. Bidder $i$'s utility is $u_i(\s;(s_i,\bids_{-i}))=v_i(\s)-v_j(s_i,\bids_{-i}).$
	
	Since $w(\bids)$ wins at $\bids$, it is known that $v_{w(\bids)}(\bids)\geq v_j(\bids)$, and by applying Corollary~\ref{col:j-bound-i}, we get $\max\{\gamma, c\} v_{w(\bids)}(s_i,\bids_{-i})\geq v_j(s_i,\bids_{-i}).$
	
	Therefore, $u_i(\s;(s_i,\bids_{-i})) \geq v_i(\s)-\max\{\gamma, c\} v_{w(\bids)}(s_i,\bids_{-i})\geq v_i(\s)- \max\{\gamma, c\} v_{w(\bids)}(\s),$ where the last inequality follows by the no-overbidding assumption and the monotonicity of the valuation functions.
	
	Thus, the proof of Theorem~\ref{prop:upper-bound:gamma-bound+c-sc:GVA} holds.
	This concludes our proof.
\end{proof}

We next show that Theorem~\ref{thm:upper-bound:gamma-bound+c-sc:2PA} is almost tight.

\begin{theorem}
	\label{pro:gamma-bound-csc:POA-close-gamma-c:2PA}
	There exist single-item settings satisfying $\gamma$-heterogeneity and $c$-SC  such that the $\npoa$ of the second-price for interdependence is arbitrarily close to $\max\{c,\gamma\}$, even under no-overbidding.\footnote{The theorem holds even under the additional assumption of Submodularity over Signals, for more information refer to \cite{EdenFFGK19}}
\end{theorem}

\begin{proof}
	The same arguments as in the proof of \cref{pro:gamma-bound-csc:POA-close-gamma-c} holds in this setting, the only change is in case 1, where the payment of agent~$1$ is $\epsilon$ instead of $1$.
\end{proof}

\section{Appendix for Section~\ref{sec:multiple-pos} (Multiple Items: A Positive Result)}
\label{appx:multiple-pos}
\subsection{Necessity of Participation Assumption}
\begin{proposition}
	There exists a multi-item, $n$ unit-demand bidders setting, satisfying single crossing and $\gamma$-homogeneity, such that the PoA of every auction that allocates each item separately to the highest valued bidder is $\Omega(m)$.
\end{proposition}

\begin{proof}
	Consider a setting with $n$ unit-demand agents and signal spaces $S_i = [0,1]$ for every agent $i$. For every item $\ell\in[m]$, let $v_{1\ell} = \sum_{j\in[n]}s_j + 1$ and let $v_{i\ell} = \sum_{j\in[n]}s_j + \frac{1}{2}$ for every $i \geq 2$.
	Any auction that allocates each item separately to the highest valued bidder, allocates all items to agent $1$, regardless of all bids. Furthermore, bidders~$\neq 1$ are indifferent of their bid, so they might as well bid their true signal. In this case, the PoA is $\Omega(m)$. 
\end{proof}

\weightedsumrevisited*

\begin{proof}[Proof]
Given a signal profile $\s$, denote the highest signal for item $\ell$ by $s_{(1)\ell}$, the second-highest signal by $s_{(2)\ell}$ and so on.
For every signal profile $\s$, consider a welfare-maximizing matching $\mu(\s)$; we can split its total value $\sum_{(i,\ell)\in \mu}v_{i\ell}(\s)$ into two components: $\sum_{(i,\ell)\in \mu}\tilde{v}_{i\ell}(\s)$ and $\sum_{(i,\ell)\in \mu}\beta\max_{j\ne i}\{s_{j\ell}\}$. The first component is upper-bounded by $\sum_{(i,\ell)\in \tildem(\s)}\tilde{v}_{i\ell}(\s)$ (where recall from \cref{def:truncated} that $\tildem(\s)$ is the maximum matching with respect to truncated values), so in expectation over $\s$ it is at most $\tildeopt$. 
It remains to bound the second component, $\sum_{(i,\ell)\in \mu}\beta\max_{j\ne i}\{s_{j\ell}\}$.

To this end, we will look at each item separately. For each item $\ell$, it holds that $\max_{j\ne i}\{s_{j\ell}\} \leq \max_{j}\{s_{j\ell}\}$.
Consider now a single item setting where each agent's value is her signal. In this setting, the above is exactly the maximal social welfare. Moreover, it is well known that when the values are drawn i.i.d from a MHR distribution, the optimal welfare is within a factor $e$ from the optimal revenue \cite[][Lemma~$3.10$]{DhangwatnotaiRY15}. Moreover, \citet{JinLQ19} show that in this setting, the second-price auction approximates the optimal revenue by a factor of $\frac{1}{1-(1-\frac{1}{e})^{n-1}}$. Note that the revenue from the second price auction is exactly $s_{(2)\ell}$.

Combining the above, we get $\beta\mathbb{E}_{\s}[\sum\limits_{(i,\ell)\in \mu}\max_{j\ne i}\{s_{j\ell}\}] \leq e(\frac{1}{1-(1-\frac{1}{e})^{n-1}})\beta\mathbb{E}_{\s}[\sum\limits_{\ell}s_{(2)\ell}]$.
As the number of agents is at least $2$ we get  $\mathbb{E}_{\s}[\sum\limits_{(i,\ell)\in \mu}\beta\max_{j\ne i}\{s_{j\ell}\}] \leq  e^2\beta \mathbb{E}_{\s}[\sum\limits_{\ell} s_{(2)\ell}]$.

Observe that for every signal profile $\s$, for every agent $i$ and item $\ell$, the truncated value $\tilde{v}_{i,\ell}(\s)$ is at least $\beta \sum_{j\ge 2} s_{(j)\ell}\ge \beta\cdot s_{(2)\ell}$. 
So the total value of $\tildem(\s)$ is at least that of the maximum matching with respect to edge weights $\beta\cdot s_{(2)\ell}$ for all edges adjacent to item $\ell$. For such edge weights, every perfect matching is maximum and has total weight $\beta\sum_{\ell}s_{(2)\ell}$. Taking expectation over $\s$ we get 
\begin{equation}
\tildeopt\ge \beta\mathbb{E}_{\s}[\sum_{\ell}s_{(2)\ell}].\label{eq:weighted-sum-re2}
\end{equation} 

Altogether we get
\begin{equation}
\opt = \sum_{(i,\ell)\in \mu}\tilde{v}_{i\ell}(\s) + \sum_{(i,\ell)\in \mu}\beta\max_{j\ne i}\{s_{j\ell}\} \le \tildeopt + e^2\beta\mathbb{E}_{\s}[\sum_{\ell}s_{(2)\ell}] \leq (1+e^2) \tildeopt .\label{eq:weighted-sum-re}
\end{equation}

As the number of agent increases, $\frac{1}{1-(1-\frac{1}{e})^{n-1}}$ is arbitrarily close to $1$, and the approximation is arbitrarily close to $(1+e)$.

\end{proof}



\multioptfartildeopt*

\begin{proof}[Proof of \cref{prop:multi-neg-for-2nd-price}]
	Consider the setting depicted in Example~\ref{ex:tildeopt-is-far}, where $F_1 = U[0,1]$, and the following strategy profile $(\bids,\avec)$: $b_{1\ell}(s_1)=\frac{1}{m}s_1$ for every $\ell$, $\avec_1(s_1)=(1,\ldots,1)$, and $b_{i\ell}(s_i)=0$ for every $\ell$, $\avec_i(s_i) = (0,\ldots,0)$ for every $i \geq 2$. That is, agent 1 is the only participant for every item. 
	
	We argue that $(\bids,\avec)$ is an ex-post equilibrium. 	
	First, note that the distribution of the signal of bidder 1 is continuous, thus there is no point mass at $0$.
	This implies $s_1 > 0$ with probability 1, and by bidder 1 strategy, $b_1 > 0$.
	Furthermore, the signal spaces of agents $i \geq 2$ are singletons, meaning they cannot bid differently than $0$.
	Keeping this in mind, for every item $\ell$, the value of bidder 1, as it seen by the mechanism, dominates the values of all other bidders (for every $\bids'$ where $b'_1>0$; $h_{1\ell}(\bids') \geq h_{1\ell}(b_{1}', {\bf 0}_{-1}) = v_{1\ell}(b_{1}', {\bf 0}_{-1}) =  v_{1\ell}(\bids') > v_{i\ell}(\bids') \geq h_{i\ell}(\bids')$).
	And by the allocation rule, if bidder 1 participates in the auction for an item, no other agent can win  by participating. 
	This follows that for every agent $i \geq 2$ there is no improving deviation.
	
	As for agent $1$, she is the only participant for every item, and by the payment rule, she pays $0$ for each item.
	Moreover, she wins all the items, thus, her utility is maximized.
	
	Finally, $\SW(\bids,\avec) = s_1$ while $\opt \geq m(1-\epsilon)s_1$.
	This concludes the proof.
\end{proof}

\section{Randomized Mechanisms}
\label{appx:future directions}


This paper focuses on simple deterministic mechanisms.
One may hope to get better bounds with respect to the broader (yet more complex) family of random mechanisms. 
This section shed some light on this topic. 


Proposition~\ref{prop:POA:RSNOSC} addresses the random-sampling Vickrey auction of \cite{EdenFFGK19} -- an ex-post IC-IR mechanism that randomly partitions the bidders into two sets $A,B$. The mechanism then allocates the item to the bidder in set~$B$ with highest ``zeroed-out'' value, where the zeroed-out value of bidder $i\in B$ is $v_i(\bs_A,s_i,{\bf 0}_{B\setminus\{i\}})$. Besides its simplicity, the big advantage of the random-sampling Vickrey auction over other auction formats in the context of interdependent values is that it achieves good welfare guarantees even with no SC assumption, as long as the SoS properties hold. To our knowledge this is the only known such mechanism to date. Unfortunately these guarantees no longer hold when we consider all equilibria and not just the truth-telling one. In fact it no longer holds even when we consider only no-overbidding equilibria, as the following proposition shows.

\begin{proposition}
	\label{prop:POA:RSNOSC}
	There exists a single-item $n$-agent setting with SoS, interdependent valuations (that do not satisfy SC), where the EP-PoA of the random-sampling Vickrey auction is $\Omega(2^n)$, even under no-overbidding.
\end{proposition}

\begin{proof}
	Consider a single-item setting with $n$ bidders, signal spaces $S_i=[0,1]$ for every $i$, and the following valuation profile: 
	\begin{equation}
	v_1 = s_1 + c \sum_{i=2}^n s_i,~~~\text{and}~~~
	\forall i\ge 2 : v_i = s_i + 2i,\label{eq:bad-example3}
	\end{equation}
	where $c$ is a sufficiently large scalar to be determined below. Observe that the valuation profile satisfies SoS. 
	We argue that the bidding strategy profile $\bids$ where $b_1(s_1) = s_1$ and for every $i\ge 2$, $b_i(s_i) = 0$  is an ex-post equilibrium.
	For every bid profile $\bb'$ and every random partition $A,B$ of the bidders chosen by the random-sampling Vickrey auction, bidder $i\ge 2$ never wins if $B$ contains a higher-indexed bidder, by the following chain of inequalities: 
	$$
	\forall i,j\in B,2\le i<j : v_i(\bb'_A,b'_i,{\bf 0}_{B-i})=b'_i+2i \le 1 + 2i < 2(i+1) \le  2j \le b'_j+2j = v_j(\bb'_A,b'_j,{\bf 0}_{B-j}).
	$$ 
	At bidding strategy profile $\bids_{-1}$, for every report $b'_1$ and every partition $A,B$, bidder $1$ wins if and only if she is the only bidder in $B$. Indeed, assuming $\text{bidder }1\in B$:
	$$
	\forall i\in B, i\ge 2: v_1({\bf 0}_A,b'_1,{\bf 0}_{B-1}) = s_1 < 2i \leq v_i({\bf 0}_A,{\bf 0}_{B}).
	$$
	Thus, whether bidder $1$ wins or not does not depend on her bid.
	We conclude that under $\bids$, the item is allocated to the highest-indexed bidder in set $B$, and no losing bidder can become a winner by deviating. Furthermore, in an ex-post IC-IR mechanism, the winner's payment according to the known characterization does not depend on her bid (see \cref{sec:challenges}). It follows that $\bids$ is an ex-post equilibrium. Also, $\bb$ clearly satisfies no-overbidding with respect to any true signal profile.
	
	Suppose now the true signal profile is $\s = (1,\dots,1)$. 
	The probability under $\bids$ that bidder~$1$ wins the item is the probability of choosing exactly bidder $1$ among all $n$ bidders to be in set $B$, i.e., $\frac{1}{2^n}$. Observe also that among all bidders $i\ge 2$, bidder $n$ has the highest value.
	The expected welfare under $\bids$ is thus at most $\frac{1}{2^n}v_1(\bs) + (1-\frac{1}{2^n})v_n(\bs) = \frac{1}{2^n}(1+c(n-1)) + (1-\frac{1}{2^n})(2n+1)$, whereas the optimal welfare obtained from always allocating the item to bidder $1$ is $1+c(n-1)$. For a sufficiently large $c$, say $c \geq 2\cdot 2^{n}$, we get $\epoa =\Omega(2^n)$, completing the proof. 
\end{proof}

In \cref{prop:POA:RSNOSC} we established a lower bound for a simple randomized auction that satisfies ex-post IC-IR. The next result (\cref{pro:BB-PNE-POA}) shows we \emph{cannot} expect a general lower bound to hold for all randomized ex-post IC-IR auctions (in contrast to deterministic such mechanisms -- recall \cref{pro:POA-EPIC}). 
Yet, our results might suggest that in order to obtain good PoA guarantees one must resort to fairly unnatural mechanisms like the one constructed in the next proposition. 


A randomized mechanism is called \emph{universally} ex-post IC-IR if it is a distribution over deterministic ex-post IC-IR mechanisms.
Informally, \cref{pro:BB-PNE-POA} shows that under the no-overbidding assumption, every $\alpha$-approximation universally ex-post IC-IR mechanism can be transformed into a randomized mechanism with $\npoa$ at most $ \frac{\alpha}{1-\epsilon}$, for arbitrarily small $\epsilon$. 
Since there exist universally ex-post IC-IR mechanisms with constant-approximation welfare guarantees for SoS valuations (such as in \cite{EdenFFGK19}),
\cref{pro:BB-PNE-POA} rules out a lower bound on the randomized mechanism's $\epoa$. The proposition makes use of a \emph{proportional allocation mechanism} defined as follows: given bid profile $\bb$, the item is allocated to bidder $i$ with probability $\frac{b_i}{\sum_jb_j + 1}$.
  
\begin{proposition}
	\label{pro:BB-PNE-POA}
	Let $\mathcal M$ be an $\alpha$-approximation welfare maximizing, universally ex-post IC-IR mechanism. For every $\epsilon>0$, the randomized mechanism $\mathcal M'$ which runs $\mathcal M$ with probability $1-\epsilon$ and runs the proportional allocation mechanism with probability~$\epsilon$ has $\npoa \le \frac{\alpha}{1-\epsilon}$ under no-overbidding.
\end{proposition}

Before proving Proposition~\ref{pro:BB-PNE-POA}, we prove the following claim.

	\begin{claim}
		\label{lem:EPIC+NOB:DSIC}
		If a mechanism  $\mathcal M$ is universally ex-post IC-IR and satisfies the no-overbidding assumption, then  $\mathcal M$ is DSIC.  
	\end{claim}
	
	\begin{proof}	
		Since every universally ex-post IC-IR mechanism is a probability distribution over deterministic ex-post IC-IR mechanisms, it is sufficient to prove the lemma for deterministic $\mathcal M$.
		Let $M=(x,p)$ be a deterministic ex-post IC-IR mechanism. Recall that $x$ is monotone and the payment rule $p$ is unique, and consist with the critical payment rule which determined only by the bids of the competing bidders (for more detail, see the characterization of ex-post IC-IR mechanisms in \cref{sec:challenges}).  
		Let $\bb$ be a bid profile. For every bidder $i$, denote $p_i(\bb_{-i})$ her payment according to $\bb$.
		
		Consider two cases:
		\begin{enumerate}
		\item $i$ wins under $\bb$, i.e., $x_i(\bb) = 1$. By monotonicity and no-overbidding, $x_i(s_i,\bb_{-i}) = 1$. It holds that $u_i((s_i, \bb_{-i});\bs) = v_i(\bs) - p_i(\vec b_{-i}) = u_i((b_i, \bb_{-i}); \bs)$.
		
		\item $i$ loses under $\bb$, i.e., $x_i(\bb) = 0$. If $x_i(s_i, \bb_{-i}) = 0$, then $u_i((s_i, \bb_{-i}); \bs) = 0 = u_i({\bb};\bs)$. Thus, it suffices to consider the scenario where $x_i(s_i, \bb_{-i}) = 1$.  By the monotonicity of $x$, the critical payment rule and by the monotonicity of the valuations we get $v_i(s_i, {\bb_{-i}}) \geq v_i(b_i^*, {\bb_{-i}}) = p_i(s_i, {\bb_{-i}})$.  Therefore ,$u_i((s_i, \bb_{-i}); \bs) = v_i(s_i, {\bb_{-i}}) - v_i(b_i^*, {\bb_{-i}}) \geq 0 = u_i(\bb; \bs)$.
		\end{enumerate}
		In both cases, $u_i((s_i, \bb_{-i}); \bs) \geq u_i(\bb;\bs)$.\\
	\end{proof}

\begin{proof}[Proof of \cref{pro:BB-PNE-POA}]
	Let $\mathcal M$ be a universally ex-post IC-IR mechanism. By \cref{lem:EPIC+NOB:DSIC}, $\mathcal M$, under no-overbidding, is in fact DSIC.
	We show that in $\mathcal M'$, truth-telling is the unique no-overbidding PNE. 
	For every bidder $i$, signal profile $\bs$ and bid profile $\bb_{-i}$, bidding $b_i=s_i$ is a best response under no-overbidding, whereas any bid $b_i < s_i$ is not a best response:
	\begin{eqnarray}
	{u_i^{\mathcal M'}((b_i,\bb_{-i});\bs)}
	&=& (1-\epsilon) {u_i^{\mathcal M}((b_i, {\bf b}_{-i});\bs)} + \epsilon\frac{b_i}{\sum_jb_j + 1} v_i(\bs) \nonumber\\
	& \leq & (1-\epsilon) {u_i^{\mathcal M}((s_i, {\bf b}_{-i});\bs)} + \epsilon\frac{b_i}{\sum_jb_j + 1} v_i(\bs) \label{eq:BB-PNE-POA:DSIC}\\
	& < & (1-\epsilon) {u_i^{\mathcal M}((s_i, {\bf b}_{-i});\bs)} + \epsilon\frac{s_i}{\sum_{j \neq i} b_j + s_i + 1} v_i(\bs) \ = \ {u_i^{\mathcal M'}((s_i,{\bf b}_{-i});\bs)},\nonumber 
	\end{eqnarray}
	where Inequality~\eqref{eq:BB-PNE-POA:DSIC} follows $\mathcal M$ is DSIC. 
	The truth-telling PNE yields at least $(1-\epsilon)/\alpha$ of the optimal welfare, and so $\npoa \le \frac{\alpha}{1-\epsilon}$.
\end{proof}



\end{document}